\documentclass[]{article}

\usepackage{amsmath, amssymb, amsthm, bbm, booktabs, color, comment, enumitem, graphicx, multicol, multirow, todonotes, url, xcolor}

\usepackage[a4paper, total={6.25in, 9.0in}]{geometry}
\usepackage[authoryear]{natbib}
\usepackage[colorlinks,citecolor=blue]{hyperref}
\usepackage[labelfont=bf]{caption}

\newcommand{\cN}{\mathcal{N}}

\newcommand{\myE}{\mathbb{E}}
\newcommand{\myP}{\mathbb{P}} 
\newcommand{\real}{\mathbb{R}} 

\newcommand{\cov}{\operatorname{cov}} 
\newcommand{\var}{\operatorname{var}}

\newcommand{\cO}{{\cal O}}
\newcommand{\dint}{\operatorname{d} \!}
\newcommand{\hsp}{\hspace{0.2mm}}
\newcommand{\hspm}{\hspace{-0.4mm}}
\newcommand{\one}{\mathbbm{1}}
\newcommand{\sgn}{\operatorname{sgn}}

\newcommand{\Fbar}{\bar{F}}
\newcommand{\Gbar}{\bar{G}}
\newcommand{\Ftilde}{\widetilde{F}}
\newcommand{\Gtilde}{\widetilde{G}}
\newcommand{\Hbar}{\bar{H}}

\newcommand{\AGC}{\operatorname{AGC}}
\newcommand{\AKC}{\operatorname{AKC}}
\newcommand{\ASC}{\operatorname{ASC}}
\newcommand{\AUC}{\operatorname{AUC}}
\newcommand{\CID}{\operatorname{CID}}
\newcommand{\CMA}{\operatorname{CMA}}
\newcommand{\CPA}{\operatorname{CPA}}
\newcommand{\RCE}{\operatorname{RCE}}

\newcommand{\twogran}{\gamma^{\hsp (2)}} 
\newcommand{\threegran}{\gamma^{\hsp (3)}} 

\newcommand{\rhoG}{\bar{\rho}}  
\newcommand{\rhoS}{\rho}  
\newcommand{\tauK}{\tau}

\theoremstyle{plain}
\newtheorem{theorem}{Theorem}[section]
\newtheorem{corollary}[theorem]{Corollary}

\newtheorem{proposition}[theorem]{Proposition}

\theoremstyle{definition}
\newtheorem{example}[theorem]{Example}
\newtheorem{remark}[theorem]{Remark}
\newtheorem{scenario}[theorem]{Scenario}

\definecolor{mycolor}{RGB}{2, 158, 115}

\begin{document}

\title{Assessing Monotone Dependence: \\ Area Under the Curve Meets Rank Correlation}
\author{Eva-Maria Walz$^{1,2}$, Andreas Eberl$^2$, Tilmann Gneiting$^{1,2}$
\vspace{0.5cm} \\ 
\small $^1$Computational Statistics Group, Heidelberg Institute for Theoretical Studies, Heidelberg, Germany \\
\small $^2$Institute of Statistics, Karlsruhe Institute of Technology (KIT), Karlsruhe, Germany}

\maketitle

\begin{abstract}
The assessment of monotone dependence between random variables is a classical problem in statistics and a gamut of application domains.  Consequently, researchers have sought measures of association that are invariant under strictly increasing transformations of the margins, with the extant literature being splintered.  For continuous variables, symmetric rank correlation coefficients, such as Spearman's Rho and Kendall's Tau, have been studied at great length in the statistical literature.  For dichotomous outcomes, the asymmetric area under the curve ($\AUC$) measure is used to assess monotone dependence.  We unify and complete thus far disconnected strands of literature, by establishing common population level theory, common estimators, and common tests that bridge continuous and dichotomous settings and apply to all linearly ordered outcomes. 

Originating in the biomedical literature, the C index provides a bridge between $\AUC$, to which it reduces for a dichotomous outcome, and Kendall's Tau, to which it relates linearly under continuity.  To establish the same kind of bridge between $\AUC$ and Spearman's Rho, we introduce asymmetric grade correlation, $\AGC(X,Y)$, as the covariance of the mid distribution function transforms, or grades, of $X$ and $Y$, divided by the variance of the grade of $Y$.  The coefficient of monotone association then is $\CMA(X,Y) = \frac{1}{2} (\AGC(X,Y) + 1)$.  The $\CMA$ measure has range $[0,1]$, and the perfect monotone predictor property holds: $\CMA(X,Y) = 1$ if, and only if, there is a nondecreasing function $m$ such that $Y = m(X)$ almost surely.  Crucially, $\CMA$ provides a bridge between $\AUC$, to which it reduces for a dichotomous outcome, and Spearman's Rho, to which it relates linearly under continuity.  Based on recent results by \citet{Pohle2025a}, we establish central limit theorems for the sample versions of these measures, and we develop tests of \citet{DeLong1988} type for their equality with a shared outcome $Y$.  In case studies, we assess progress in data-driven weather prediction and evaluate methods of uncertainty quantification for large language models.

\medskip
\textit{Keywords}: area under the curve, grade correlation, Kendall's Tau, concordance index, coefficient of monotone association, DeLong test, Somers' $D$, Spearman's Rho
\end{abstract}

\section{Introduction}  \label{sec:introduction}

There is a vast extant literature on quantitative measures of dependence between real-valued random variables $X$ and $Y$ \citep{Embrechts2002}.  Generally, one distinguishes measures of complete dependence \citep[e.g.,][]{Renyi1959, Neslehova2007, Szekely2007, Reshef2011} from measures of concordance or monotone dependence \citep[e.g.,][]{Schweizer1981, Scarsini1984}.  Typically, such measures are symmetric in the sense that their value remains the same if we interchange the roles of $X$ and $Y$.  Despite the focus on symmetric measures of general or monotone dependence in the statistical literature, we follow \citet{Chatterjee2021} in this paper and argue that asymmetric measures, which honor the specific roles of $X$ as a covariate, feature, or marker, and of $Y$ as the outcome of interest, have distinct benefits.  A key argument in favor of asymmetric measures is that, while independence is a symmetric property of $X$ and $Y$, the classical notions of complete dependence \citep{Lancaster1963} and perfect monotone dependence \citep{Kimeldorf1978} fail to be symmetric in discrete settings.

Specifically, consider the pair $(X,Y)$ of real-valued random variables with joint distribution $\myP$.  We generally assume that neither the covariate $X$ nor the outcome $Y$ is almost surely constant and refer to this setting as nondegenerate.  Let $F(x) = \myP(X \leq x$) denote the cumulative distribution function (CDF) of $X$, and let $G$ denote the CDF of $Y$.  When $F$ and $G$ are continuous, Spearman's rank correlation coefficient \citep{Spearman1904} is defined as 
\begin{align}  \label{eq:rhoS}
\rhoS(X,Y) = 12 \hsp \cov(F(X),G(Y)),  
\end{align}
which equals the Pearson product moment correlation between the probability integral transforms $F(X)$ and $G(Y)$.  Kendall's rank correlation coefficient \citep{Kendall1938} is defined as 
\begin{align}  \label{eq:tau}
\tauK(X,Y) = \myE \left( \sgn(X' - X'') \sgn(Y' - Y'') \right) \! ,
\end{align}
where here and in the remainder of the paper $(X,Y)$, $(X',Y')$, and $(X'',Y'')$ have distribution $\myP$ and are independent.  The classical rank correlation coefficients are symmetric, and we often write
\begin{align*}
\rhoS = \rhoS(X,Y) = \rhoS(Y,X) \qquad \textrm{and} \qquad \tauK = \tauK(X,Y) = \tauK(Y,X), 
\end{align*}
respectively.  Furthermore, Spearman's Rho and Kendall's Tau are measures of monotone dependence in the sense that they are invariant under strictly increasing transformations of the margins, and thus they depend on only the copula when $X$ and $Y$ are continuous.\footnote{When $X$ and $Y$ are both continuous, the copula of the bivariate distribution $\myP$ is the bivariate distribution of $(F(X), G(Y))$.  As $F$ is strictly increasing on the support of $X$, and $G$ is strictly increasing on the support of $Y$, it follows that a measure of monotone dependence can be expressed in terms of the copula only \citep{Schweizer1981}.}  These and other symmetric measures of monotone dependence can readily be applied to arbitrary distributions, without requiring continuity.  However, generalizations to arbitrary distributions are subject to an attainability problem, as for noncontinuous distributions the maximal value of the measure may not be reached even though a perfect monotone relationship holds \citep{Neslehova2007, Mesfioui2022}.  We demonstrate in the subsequent section that one needs to resort to asymmetric measures in order to resolve these issues.

The most prominent and most pronounced case of discreteness arises under a dichotomous or binary outcome $Y$.  We adopt a common convention in the dichotomous case and refer to $Y = 1$ as a positive outcome, and to $Y = 0$ as a negative outcome.  In this scenario, the receiver operating characteristic (ROC) curve and the area under the ROC curve ($\AUC$) measure \citep{Swets1988, Bradley1997, Fawcett2006} are widely used tools for the assessment of monotone dependence of $Y$ on a real-valued covariate, feature, or marker $X$.  Specifically, if $Y = Y_d$ is a nondegenerate dichotomous outcome, let $F_1$ and $F_0$ denote the cumulative distribution function (CDF) of $X$ conditional on $Y_d = 1$ and $Y_d = 0$, respectively.  The ROC curve is a plot of the hit rate $1 - F_1(x)$ versus the false alarm rate $1 - F_0(x)$ as the decision threshold $x$ varies.  It is well known that the area under the ROC curve ($\AUC$) equals 
\begin{align}  \label{eq:AUC} 
\AUC(X,Y_d) = \myE \left( s(X',X'') \mid Y_d' = 0, Y_d'' = 1 \right) \! ,
\end{align}
where
\begin{align}  \label{eq:s}
s(x',x'') = \one \{ x' < x'' \} + \frac{1}{2} \one \{ x' = x'' \} 
\end{align}
for $x', x'' \in \real$; for technical details see, e.g., \citet{Gneiting2022a}.  Traditionally, ROC analysis and the $\AUC$ measure have been limited to dichotomous outcomes.  However, real-valued variables are ubiquitous, and researchers have felt compelled to dichotomize continuous outcomes so that these popular tools can be applied, with substantial adverse effects on analyses \citep{Altman2006, Huang2024}.  

Our goal in this paper is to show that thus far disconnected strands of literature, with disparate theories and methodologies for the continuous case and dichotomous outcomes, respectively, can be unified.  Specifically, we establish common population level theory, common estimators, and common tests that bridge continuous and dichotomous settings and generalize to all linearly ordered outcomes.  To describe the first such bridge, consider the C index ($\CID$) of $Y$ on $X$ as introduced by \citet{Harrell1996} and studied by \citet{Pencina2004}, 
\begin{align*}
\CID(X,Y) = \myE \left( s(X',X'') \mid Y' < Y'' \right) ,
\end{align*}
which is the tie-adjusted probability of concordance.  Evidently, if $Y = Y_d$ is dichotomous then $\CID(X,Y_d) = \AUC(X,Y_d)$, whereas if $X = X_c$ and $Y = Y_c$ are continuous then $\CID(X_c,Y_c) = \frac{1}{2} (\tauK + 1)$ relates linearly to Kendall's Tau.  

To establish the same kind of bridge between $\AUC$ and Spearman's Rho, let $\Fbar(x) = \myP(X < x) + \frac{1}{2} \, \myP( X = x)$ denote the mid distribution function (MDF) of $X$.  Similarly, let $\Gbar$ denote the MDF of $Y$.  We define the \textit{asymmetric grade correlation}~($\AGC$) of $Y$ on $X$ as 
\begin{align*}
\AGC(X,Y) = \frac{\cov(\Fbar(X), \Gbar(Y))}{\cov(\Gbar(Y), \Gbar(Y))},
\end{align*}
and the \textit{coefficient of monotone association}~($\CMA$) of $Y$ on $X$ as 
\begin{align*}
\CMA(X,Y) = \frac{1}{2} \left( \AGC(X,Y) + 1 \right) \! ,  
\end{align*}
respectively.  Perhaps surprisingly, despite the asymmetry in the defining relation, $\AGC(X,Y) \in [-1,1]$ and $\CMA(X,Y) \in [0,1]$ are normalized.  Furthermore, the upper bound is attained if, and only if, there is a nondecreasing function $m$ such that $Y = m(X)$ almost surely.  When $Y = Y_d$ is dichotomous, then $\CMA(X,Y_d) = \AUC(X,Y_d)$, whereas if $X_c$ and $Y_c$ are continuous then $\CMA(X_c,Y_c) = \frac{1}{2} (\rhoS + 1)$ relates linearly to Spearman's Rho.  In this sense, $\CMA$ bridges Spearman's Rho and $\AUC$ in the very same way that $\CID$ bridges Kendall's Tau and $\AUC$.

The remainder of the paper is organized as follows.  The subsequent Section~\ref{sec:population} elaborates the, in part surprising, properties of the C index and the $\CMA$ measure at the population level.  Section~\ref{sec:sample} turns to empirical settings, where we study plug-in sample versions of the measures and draw on the recent work of \citet{Pohle2025a} to devise tools for inference that cover the full range from dichotomous to continuous outcomes.  In particular, we establish central limit theorems for the sample versions of the measures, and we develop general tests of \citet{DeLong1988} type for the comparison of monotone association with a shared outcome, which nest classic tests for the equality of $\AUC$, Spearman's Rho, and Kendall's Tau, respectively.  Section~\ref{sec:case_studies} presents case studies on revolutionary developments in the recent transition from physics-based to data-driven weather prediction, and on the assessment of uncertainty statements for large language models.  The paper closes with a discussion in Section~\ref{sec:discussion}.  Proofs and technical details are deferred to appendices.

\section{Bridging the area under the curve (AUC) measure to rank correlation coefficients: Population setting}  \label{sec:population}

We now provide details for the two bridges between the classical rank correlation coefficients and the area under the curve ($\AUC$) measure.  In a nutshell, the C index ($\CID$) is normalized to the unit interval $[0,1]$, reduces to $\AUC$ when the outcome is dichotomous, and specializes to a rescaled version of Kendall's Tau when the variables are continuous.  We then introduce the coefficient of monotone association ($\CMA$), which also is normalized to the unit interval $[0,1]$, reduces to $\AUC$ when the outcome is dichotomous, and specializes to a rescaled version of Spearman's Rho when the variables are continuous.  We operate in the population setting, where the bivariate random vector $(X,Y)$ has distribution $\myP$, and we assume that $X$ and $Y$ are nondegenerate, in the sense that neither of them is almost surely constant.

\subsection{The C index bridges Kendall's Tau and AUC}  \label{sec:CID}
 
The \textit{C index}\/ or \textit{concordance index}\/ of $X$ on $Y$ is defined as
\begin{align}  \label{eq:CID}
\CID(X,Y) = \myE \left( s(X',X'') \mid Y' < Y'' \right) \! ,
\end{align}
where the function $s(x',x'') = \one \{ x' < x'' \} + \frac{1}{2} \one \{ x' = x'' \}$ is specified at \eqref{eq:s}.  Of course, $\CID(X,Y)$ is the probability of concordance between $X$ and $Y$, with a natural adjustment for ties in $X$, which renders the measure asymmetric.  In the nondegenerate case both $\CID(X,Y)$ and $\CID(Y,X)$ are well defined.  The C index was developed by \citet{Harrell1996} and \citet{Pencina2004} and has traditionally been applied to censored outcomes in survival analysis.  Here, we ignore the various complications that arise under censoring \citep{Hartman2023, Sierra2025, Lillelund2026}\footnote{Under censoring one does not observe the outcome $Y$ of interest, but instead a censored variable $C = \min(Y,Z)$ along with an indicator $\delta = \one \{ Y = C \}$.  The challenge is that one needs to estimate $\CID(X,Y)$ based on the triple $(X, C, \delta)$.  While the handling of censored outcomes and/or missing data is of great importance, we here restrict attention to the nested setting where the outcome $Y$ of interest is noncensored and fully observed.} and assume that the outcome $Y$ is fully observed.

For further analysis here and in subsequent sections, we consider the notion of granularity.  Specifically, for a random variable $Y$, the 2-granularity $\twogran(Y)$ and the 3-granularity $\threegran(Y)$ are defined as 
\begin{align}  \label{eq:granularity} 
\twogran(Y) = \myP \left( Y = Y' \right) \quad \textrm{and} \quad \threegran(Y) = \myP \left( Y = Y' = Y'' \right) \! ,
\end{align}
respectively.  Evidently, $0 \leq \threegran(Y) \leq \twogran(Y) \leq 1$.  The minimal granularity arises under the finest possible resolution, namely, for continuous distributions, and the maximum is reached for degenerate distributions with all mass in a single atom. 

In the nondegenerate case, we define the \textit{asymmetric Kendall correlation}~($\AKC$) of $Y$ on $X$ as 
\begin{align}  
\AKC(X,Y) & = \frac{\tauK(X,Y)}{\tauK(Y,Y)} \label{eq:AKC} \\
          & = \frac{\tauK}{1 - \twogran(Y)} \label{eq:AKC_granularity} 
\end{align}
in terms of Kendall's Tau at \eqref{eq:tau}.  This measure has been studied under the label of Somers' $D$ in the social science literature \citep{Somers1962, Newson2002}.  We readily see that
the C index, 
\begin{align}  \label{eq:CID_AKC} 
\CID(X,Y) = \frac{1}{2} \left( \AKC(X,Y) + 1 \right) ,
\end{align}
relates linearly to $\AKC(X,Y)$.  When $Y = Y_d$ is dichotomous then 
\begin{align}  \label{eq:CID_dichotomous}
\CID(X,Y_d) = \AUC(X,Y_d)
\end{align}
reduces to the $\AUC$ measure at \eqref{eq:AUC}.  On the other hand, if $X = X_c$ and $Y = Y_c$ are continuous then $\AKC(X_c,Y_c) = \AKC(Y_c,X_c) = \tauK$ and
\begin{align}  \label{eq:CID_continuous}
\CID(X_c,Y_c) = \CID(Y_c,X_c) = \frac{1}{2} \left( \hsp \tauK + 1 \right) .
\end{align}
Summarizing the relationships at \eqref{eq:CID_continuous} and \eqref{eq:CID_dichotomous}, the $\CID$ measure bridges Kendall's Tau, to which it relates linearly in the continuous case, and $\AUC$, to which it reduces when the outcome is dichotomous.  The subsequent simple example serves to illustrate the adaptation of $\AKC(X,Y)$ and $\CID(X,Y)$ to the granularity of $X$ and $Y$. 

\begin{example}  \label{ex:granularity_akc}
Let $X$ be uniform on $(0,1]$, let the function $m$ be strictly increasing, and let the outcome
\begin{align*} 
Y_k = \sum_{i=1}^k \one \! \left\{ X \in \left( \frac{i-1}{k}, \frac{i}{k} \right] \right\}  m \! \left( \frac{i}{k} \right)
\end{align*}
be discrete, where $k \geq 2$ is an integer.  Then $\tauK(X,Y_k) = \tauK(Y_k,X) = 1 - \frac{1}{k}$, whereas $\twogran(X) = 0$ and $\twogran(Y_k) = \frac{1}{k}$.  Therefore,
\begin{align*}
\AKC(X,Y_k) = 1 \quad \textrm{and} \quad \AKC(Y_k,X) = 1 - \frac{1}{k},
\end{align*}
respectively.  The value of $\AKC(X,Y_k) = 1$ reflects the fact that the continuous covariate $X$ is a perfect monotone predictor of the discrete outcome $Y_k$.  Likewise, the value of $\AKC(Y_k,X) = 1 - \frac{1}{k}$ quantifies the suboptimality of $Y_k$ as a predictor of $X$.  In the limit as $k \to \infty$ the 2-granularities of $X$ and $Y_k$ become equal and the difference between $\AKC(X,Y_k)$ and $\AKC(Y_k,X)$ vanishes.  Evidently, analogous relations hold for $\CMA(X,Y_k)$ and $\CMA(Y_k,X)$.
\end{example}

There is an expansive literature on desirable properties for measure of complete dependence \citep{Renyi1959} and monotone dependence \citep{Schweizer1981, Scarsini1984}.  We now summarize key properties of $\AKC$ that carry over to the $\CID$ measure in obvious ways. 

\begin{theorem}  \label{thm:AKC_properties}
Let\/ $(X,Y)$ be nondegenerate. 
\begin{itemize}
\item[(a)] Range: $- 1 \leq \AKC(X,Y) \leq 1$
\item[(b)] Symmetry under equal 2-granularity: $\AKC(X,Y) = \AKC(Y,X)$ if, and only if, $\twogran(X) = \twogran(Y)$
\item[(c)] Change of sign: $\AKC(-X,Y) = \AKC(X,-Y) = - \AKC(X,Y)$
\item[(d)] Independence: If\/ $X$ and\/ $Y$ are independent, then\/ $\AKC(X,Y) = 0$.
\item[(e)] Perfect monotone predictor property: $\AKC(X,Y) = 1$ if, and only if, there exists a nondecreasing function\/ $m$ such that\/ $Y = m(X)$ almost surely.
\item[(f)] Invariance under strictly increasing transformations: If\/ $f$ and\/ $g$ are strictly increasing, then\/ $\AKC( \hsp f(X),g(Y)) = \AKC(X,Y)$. 
\end{itemize}
\end{theorem}

While (a), (c), (d), and (f) are well established properties and concepts, we note the modulation of symmetry in terms of 2-granularity in property (b).  The perfect monotone predictor property (e) states that $\AKC(X,Y) = 1$ if, and only if, $X$ is a perfect monotone predictor of $Y$ in the sense of \citet{Kimeldorf1978}.  Therefore, $\AKC$ and $\CID$ overcome the attainability problem.  Furthermore, $\AKC$ and $\CID$ are particularly relevant when one seeks to assess the extent to which an explanatory variable $X$ might contribute to the prediction of an outcome $Y$.  The perfect monotone predictor property is violated by Kendall's Tau and Spearman's Rho, and indeed it is incompatible with the classical notion of symmetry for a measure of association.  The simple reason is that, if $X$ is a perfect monotone predictor of $Y$, then $Y$ may or may not be a perfect monotone predictor of $X$, as elucidated in the following comment.

\begin{remark}[incompatibility of symmetry and perfect monotone predictor property]  \label{re:incompatibility}
If $X$ and $Y$ are both continuous, then $X$ is a perfect monotone predictor of $Y$ if, and only if, $Y$ is a perfect monotone predictor of $X$ \citep[p.~897]{Kimeldorf1978}.  However, if we drop the assumption of continuity, counterexamples to the equivalence arise.  For instance, for the pair $(X,Y_k)$ from Example \ref{ex:granularity_akc} the continuous variate $X$ is a perfect monotone predictor of the discrete outcome $Y_k$, but $Y_k$ fails to be a perfect monotone predictor of $X$.  Consequently, it is not possible to find a symmetric measure of association with the perfect monotone predictor property. 
\end{remark}

Interestingly, the perfect monotone predictor property is compatible with the modulated symmetry condition (b).  Furthermore, jointly with the change of sign property (d), the perfect monotone predictor property implies that $\AKC(X,Y) = - 1$ if, and only if, there is a nonincreasing function $m$ such that $Y = m(X)$ almost surely.  Evidently, analogous statements hold for $\CID(X,Y)$.

\subsection{The coefficient of monotone association (CMA) measure bridges Spearman's Rho and AUC}  \label{sec:AGC}

We now aim to bridge Spearman's Rho and $\AUC$ in the same way that the C index bridges Kendall's Tau and $\AUC$.  Specifically, we introduce the coefficient of monotone association ($\CMA$), which is normalized to the unit interval $[0,1]$, reduces to $\AUC$ when the outcome is dichotomous, and specializes to a rescaled version of Spearman's Rho when the variables are continuous.  To this end, it is tempting to follow the route in the previous section and depart from an asymmetric version of Spearman's Rho.\footnote{Recall that $\AKC(X,Y) = \tauK(X,Y)/\tau(Y,Y)$ and $\CID(X,Y) = \frac{1}{2} (\AKC(X,Y) + 1)$.  In this light, it may be tempting to define an asymmetric Spearman's rank correlation coefficient as $\ASC(X,Y) = \rhoS(X,Y)/\rhoS(Y,Y)$ and rescale to $\frac{1}{2} (\ASC(X,Y) + 1)$.  However, we have not been able to prove an analogue of the crucial Theorem \ref{thm:CMA2s} in this alternate setting.}  However, such an approach is subject to the technical challenge of analyzing the standard version of Spearman's Rho in discrete settings \citep{Mesfioui2022}.  We address this issue by using an alternate version of Spearman's Rho, in accordance with the classical paper by \citet[pp.~317--319]{Hoeffding1948} and Definition 2.3 in the recent work of \citet{Pohle2025a}.

Specifically, let $F(x) = \myP(X \leq x)$ be the cumulative distribution function (CDF) of the random variable $X$, which we typically interpret as covariate, feature, or marker, and let $G$ denote the CDF of the outcome $Y$.  Let $\Fbar(x) = \myP( X < x) + \frac{1}{2} \, \myP(X = x)$ denote the mid distribution function (MDF) of $X$.  Similarly, let $\Gbar$ denote the MDF of $Y$.  The random variables $\Fbar(X)$ and $\Gbar(Y)$ are called the grades\footnote{The term grade can be traced to \citet[pp.~318--319]{Hoeffding1948} and is meant in the sense of rank.  An alternative label is mid distribution transform as used by \citet{Parzen2004} and \citet{Ma2011}.} of $X$ and $Y$, and the symmetric grade correlation
\begin{align}  \label{eq:rhoG}
\rhoG = 12 \hsp \cov(\Fbar(X),\Gbar(Y))
\end{align}
is the desired alternate version of Spearman's Rho at \eqref{eq:rhoS}.  Of course, if $X = X_c$ and $Y = Y_c$ are continuous, then $\Fbar = F$ and $\Gbar = G$ so that the alternate version agrees with the standard version. 

In the general setting, we introduce the \textit{asymmetric grade correlation}~($\AGC$) of $Y$ on $X$ as 
\begin{align}
\AGC(X,Y) & = \frac{\cov(\Fbar(X), \Gbar(Y))}{\cov(\Gbar(Y), \Gbar(Y))} \label{eq:AGC} \\
          & = \frac{\rhoG}{1 - \threegran(Y)}, \label{eq:AGC_granularity}
\end{align}
where 3-granularity $\threegran$ is defined at \eqref{eq:granularity} and the equality of the representations follows from results in \citet{Niewiadomska2005} and \citet{Pohle2025a}.  The associated \textit{coefficient of monotone association}~($\CMA$) measure, 
\begin{align}  \label{eq:CMA}
\CMA(X,Y) = \frac{1}{2} \left( \AGC(X,Y) + 1 \right), 
\end{align}
is an affine function of $\AGC$.  

Perhaps surprisingly, the following representation of $\CMA$ and $\AGC$ in terms of the similarity function $s$ at \eqref{eq:s} establishes that the measures are normalized.  Specifically, $\CMA$ has range $[0,1]$ and $\AGC$ attains values in the interval $[-1,1]$.

\begin{theorem}  \label{thm:CMA2s} 
For nondegenerate\/ $X$ and\/ $Y$, it holds that
\begin{align}
\CMA(X,Y) = \frac{\myE \left[ (\Gbar(Y'') - \Gbar(Y')) \hsp \one \{ Y' < Y'' \} \hsp s(X',X'') \right]}
            {\myE \left[ (\Gbar(Y'') - \Gbar(Y')) \hsp \one \{ Y' < Y'' \} \right]}. \label{eq:CMA2s}
\end{align}
\end{theorem}

In addition to ensuring the desired normalization, the representation \eqref{eq:CMA2s} demonstrates that $\CMA$ and $\AGC$ are measures of the amount of probability mass between outcomes that order in the same way as the covariate.  If there are ties in the covariate, the corresponding probability mass is discounted by a factor of one half. 

Crucially, the $\CMA$ measure bridges Spearman's Rho and $\AUC$ in the same way that the C index bridges Kendall's Tau and $\AUC$.  If $Y = Y_d$ is dichotomous then \eqref{eq:CMA2s} implies that
\begin{align}  \label{eq:CMA_dichotomous}  
\CMA(X,Y_d) = \AUC(X,Y_d).
\end{align}
On the other hand, when $X = X_c$ and $Y = Y_c$ are continuous then $\AGC(X_c,Y_c) = \AGC(Y_c,X_c) = \rhoG = \rhoS$ and
\begin{align}  \label{eq:CMA_continuous}
\CMA(X_c,Y_c) = \CMA(Y_c,X_c) 
= \frac{1}{2} \left( \hsp \rhoG + 1 \right) 
= \frac{1}{2} \left( \hsp \rhoS + 1 \right) .
\end{align}

In what follows, let $Y_c$ be a continuous random variable with CDF $G$.  For $\alpha \in (0,1)$, we define
\begin{align}  \label{eq:AUC_alpha} 
\AUC^{\hsp (\alpha)}(X,Y_c) = \AUC(X, \one \{ Y_c \geq G^{-1}(\alpha) \})
\end{align}
as the $\AUC$ measure \eqref{eq:AUC} when the continuous outcome $Y_c$ is dichotomized at the $\alpha$-quantile of its marginal distribution.

\begin{theorem}  \label{thm:CMA2AUC} 
If\/ $Y_c$ is continuous, it holds that
\begin{align}  \label{eq:CMA2AUC}
\CMA(X,Y_c) = 6 \int_0^1 \alpha \left( 1 - \alpha \right) \AUC^{\hsp (\alpha)}(X,Y_c) \dint \alpha. 
\end{align} 
\end{theorem}

The representation \eqref{eq:CMA2AUC} was hinted at by \citet[equation (21)]{Gneiting2022b} in the special case when both variables are continuous, to which we tend now.

\begin{corollary}  \label{cor:CMA2AUC}
If\/ $X_c$ and\/ $Y_c$ are continuous, then  
\begin{align} 
\CMA(X_c,Y_c) 
& = 6 \int_0^1 \alpha \left( 1 - \alpha \right) \AUC^{\hsp (\alpha)}(X_c,Y_c) \dint \alpha \label{eq:CMA2AUC_continuous} \\ 
& = \frac{1}{2} \left( \hsp \rhoG + 1 \right) = \frac{1}{2} \left( \hsp \rhoS + 1 \right) \nonumber \\
& = 6 \int_0^1 \alpha \left( 1 - \alpha \right) \AUC^{\hsp (\alpha)}(Y_c,X_c) \dint \alpha = \CMA(Y_c,X_c). \label{eq:CMA2AUC_symmetric}
\end{align} 
\end{corollary}

\begin{example}  \label{ex:CMA2AUC}
Let $(X,Y)$ be uniform on the triangle $\textrm{T}$ with vertices $(0,0)$, $(\frac{1}{2},1)$, and $(1,0)$ in the Euclidean plane.  Evidently, $X$ and $Y$ are dependent.  However, the dependence of $Y$ on $X$ is not monotone; e.g., if $x < \frac{1}{2}$ then $X \leq x$ implies $Y < 2x$, but $X > 1 - x$ also implies $Y < 2x$.  On the other hand, the conditional distribution of $X$ given $Y = y \in [0,1]$ is uniform on the interval $[\frac{y}{2}, 1 - \frac{y}{2}]$.  Not surprisingly, Spearman's Rho and Kendall's Tau vanish.  In Appendix \ref{app:proofs_population} we show that 
\begin{align*}
\AUC^{\hsp (\alpha)}(X,Y) = \frac{1}{2}
\end{align*}
for $\alpha \in (0,1)$, whereas 
\begin{align*} 
\AUC^{\hsp (\alpha)}(Y,X) = 1 - \frac{\frac{2}{3} \sqrt{2\alpha} + \frac{\alpha}{6}}{1 - \alpha} 
\end{align*} 
and $\AUC^{\hsp (1-\alpha)}(Y,X) = 1 - \AUC^{\hsp (\alpha)}(Y,X)$ for $\alpha \in (0,\frac{1}{2})$.  The common value of $\CMA(X,Y)$ and $\CMA(Y,X)$ in \eqref{eq:CMA2AUC_continuous} and \eqref{eq:CMA2AUC_symmetric}, respectively, is $\frac{1}{2}$, in accordance with $\rhoS = 0$.
\end{example}

Next we summarize key properties of $\AGC$ that carry over to the $\CMA$ measure in obvious ways.

\begin{theorem}  \label{thm:AGC_properties}
Let\/ $(X,Y)$ be nondegenerate. 
\begin{itemize}
\item[(a)] Range: $- 1 \leq \AGC(X,Y) \leq 1$
\item[(b)] Symmetry under equal 3-granularity: $\AGC(X,Y) = \AGC(Y,X)$ if, and only if, $\threegran(X) = \threegran(Y)$
\item[(c)] Change of sign: $\AGC(-X,Y) = \AGC(X,-Y) = - \AGC(X,Y)$
\item[(d)] Independence: If\/ $X$ and\/ $Y$ are independent, then\/ $\AGC(X,Y) = 0$. 
\item[(e)] Perfect monotone predictor property: $\AGC(X,Y) = 1$ if, and only if, there exists a nondecreasing function\/ $m$ such that\/ $Y = m(X)$ almost surely.
\item[(f)] Invariance under strictly increasing transformations: If\/ $f$ and\/ $g$ are strictly increasing, then\/ $\AGC( \hsp f(X),g(Y)) = \AGC(X,Y)$. 
\end{itemize}
\end{theorem}

For the proof see Appendix \ref{app:proofs_population}.  While the normalization in (a) follows readily from the particular representation \eqref{eq:CMA2s}, we have not been able to find a more direct and more intuitive argument.  Except for the modulated symmetry property, which is now in terms of 3-granularity, the attributes in Theorem \ref{thm:AKC_properties} are the same as in Theorem \ref{thm:AKC_properties}.  In particular, $\AGC$ and $\CMA$ have the perfect monotone predictor property and overcome the attainability problem.  Finally, we note that the quantity $\AGC(X,Y)$ arises as the slope in rank-rank regression models \citep{Chetverikov2023}, provided that the ranks take the form of mid ranks.

\subsection{Which measure to choose?}  \label{sec:choice}

Evidently, the choice of the normalization of the measure --- that is, the choice of $\AKC$ versus $\CID$, or $\AGC$ versus $\CMA$ --- is a matter of taste.  A plausible preference for $\CID$ and $\CMA$ is in application areas where researchers have artificially dichotomized outcomes in order to use $\AUC$.  Likewise, $\AKC$ and $\AGC$ are attractive in fields where traditionally symmetric rank correlation coefficients have been used, as they resolve the attainability problem for discrete outcomes, while reducing to the standard measures for continuous variables.

\begin{figure}[t]
\centering
\includegraphics[width = 0.7 \textwidth]{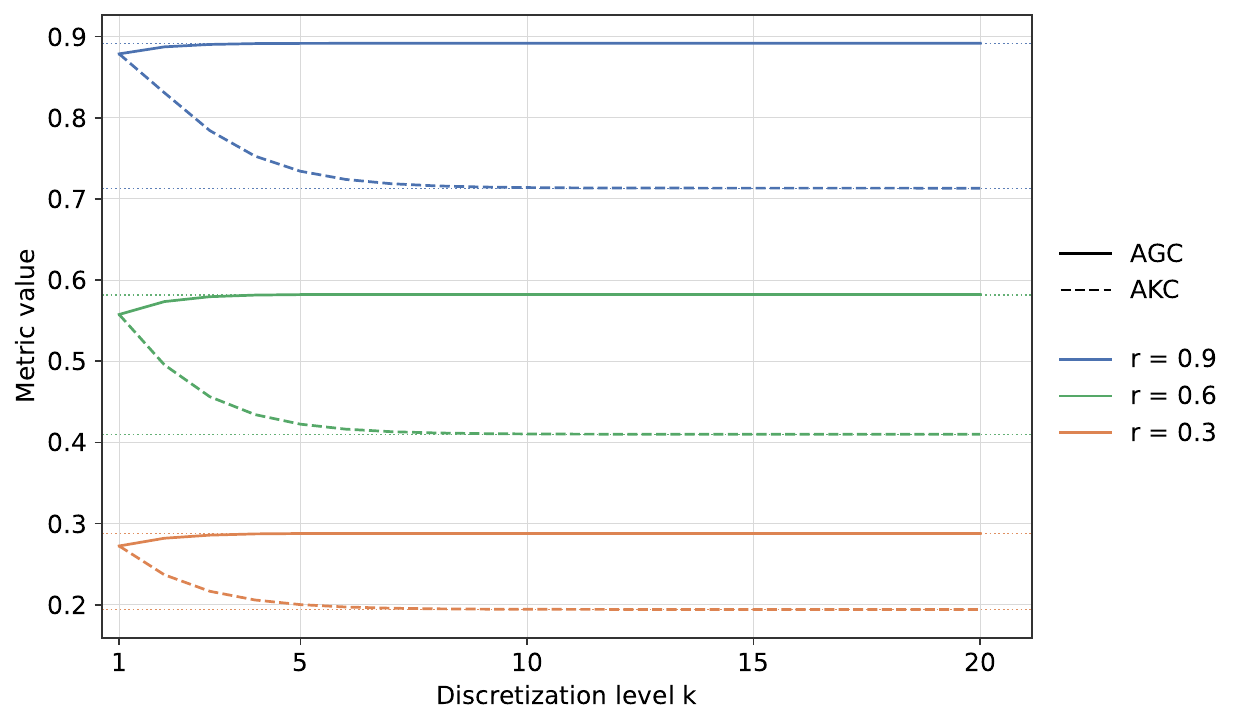}
\caption{The quantities $\AGC(X,Y_k)$ and $\AKC(X,Y_k)$, where $(X,Y)$ is bivariate normal with Pearson correlation $r$ and $Y$ is discretized into an ordinal outcome $Y_k$ with $2^k$ equi-probable values.  The dotted lines show $\rhoS(X,Y)$ and $\tauK(X,Y)$, respectively.  \label{fig:AGC_AKC}}
\end{figure}

In the remainder of the section, we address the more subtle question of the choice between $\AKC$ and $\AGC$, with the arguments applying equally to the choice between $\CID$ and $\CMA$.  For a dichotomous outcome $Y_d$, the measures are identical, in that $\AGC(X,Y_d) = \AKC(X,Y_d) = 2 \AUC(X,Y_d) - 1$, regardless of class balance.  For continuous variables, the choice between $\AKC$ and $\AGC$ reduces to the choice between Kendall`s Tau and Spearman's Rho, for which there is an impressive literature with comprehensive insights and results \citep{Daniels1950, Kruskal1958, Schreyer2017}.  For many continuous distributions with weak dependencies, the ratio of Spearman's Rho over Kendall's Tau equals about three over two \citep{Fredricks2007}.  In particular, for a bivariate Gaussian random vector $(X,Y)$ with Pearson correlation $r \in (0,1)$, it holds that
\begin{align*}
\AGC(X,Y) = \AGC(Y,X) = \rhoG = \rhoS = \frac{6}{\pi} \arcsin \frac{r}{2} < \frac{2}{\pi} \arcsin r = \tauK = \AKC(Y,X) = \AKC(X,Y).  
\end{align*}

Generally, the difference between $\AKC(X,Y)$ and $\AGC(X,Y)$ depends strongly on the granularity of the outcome.  We illustrate this dependence in a setting akin to Section~4.4 and Figure 5 in \citet{Gneiting2022b}.  Specifically, let $(X,Y)$ be bivariate Gaussian with Pearson correlation $r \in \{ 0.3, 0.6, 0.9 \}$, respectively, and discretize $Y$ into an ordinal outcome $Y_k$ with $2^k$ equi-probable values.  Figure \ref{fig:AGC_AKC} shows $\AGC(X,Y_k)$ and $\AKC(X,Y_k)$ as $k$ ranges from 1 to 20.  When $k = 1$, the outcome is dichotomous and we note that $\AGC(X,Y_1) = \AKC(X,Y_1) = 2 \AUC(X,Y_1) - 1$.  As $k$ grows, $\AGC(X,Y_k)$ increases toward $\rhoS(X,Y)$, whereas $\AKC(X,Y_k)$ decreases toward $\tauK(X,Y)$.  In line with the representations at \eqref{eq:AKC_granularity} and \eqref{eq:AGC_granularity} and the fact that $\threegran(Y_k) \geq \twogran(Y_k)$, $\AGC(X,Y_k)$ depends much less on the discretization level $k$ of the outcome than $\AKC(X,Y_k)$.  

These findings inform the choice between the measures.  In particular, if researchers are looking for results that depend little on the discretization level of their outcomes --- a factor that may be arbitrary and beyond their influence --- then they may prefer $\AGC$ over $\AKC$.  A further benefit of $\AGC$ and $\CMA$ is the availability of the mixture representation for $\CMA(X,Y)$ at \eqref{eq:CMA2AUC}, which strengthens the link to the popular $\AUC$ measure.  We are unaware of any such representation for the $\AKC$ and $\CID$ measures.

While $\AKC$, $\CID$, $\AGC$, and $\CMA$ are asymmetric measures of monotone association, a plethora of symmetric measures of monotone association have been developed, such as the maximum information coefficient \citep{Reshef2011}, Bergsma--Dassios sign correlation \citep{Bergsma2014}, symmetric rank correlations \citep{Weihs2018}, Hellinger correlation \citep{Geenens2022}, and generalized correlations \citep{Fissler2023}.  The symmetric measures generally satisfy properties (a), (c), (d), and (f) from Theorems \ref{thm:AKC_properties} and \ref{thm:AGC_properties}, but due to Remark \ref{re:incompatibility} the perfect monotone predictor property (e) must fail.  In this light, we refrain from comparisons to symmetric measures.  

However, we discuss a number of striking commonalities and differences between $\AKC$, $\CID$, $\AGC$, and $\CMA$, respectively, and the recently proposed, asymmetric coefficient $\xi(X,Y)$ of \citet{Chatterjee2021} in Appendix \ref{app:Chatterjee}.  Interestingly, while our coefficients become symmetric when the variables are continuous, we note an example where $\xi(X,Y) \not= \xi(Y,X)$ even though $X$ and $Y$ are continuous.

\section{Sampling theory and tests of DeLong type}  \label{sec:sample}

Following the lead of \citet{Pohle2025a}, we define empirical versions of the $\AKC$, $\CID$, $\AGC$, and $\CMA$ measures in a simple yet efficient way.  Specifically, we plug the empirical law, $\myP_n$, of data 
\begin{align}  \label{eq:data} 
(X_1,Y_1), \ldots, (X_n,Y_n)
\end{align} 
into the population representations at \eqref{eq:AKC}, \eqref{eq:CID_AKC}, \eqref{eq:AGC}, and \eqref{eq:CMA},  to yield empirical coefficients, which we denote by $\AKC_n$, $\CID_n$, $\AGC_n$, and $\CMA_n$, respectively. Building on the classical asymptotic theory of U-statistics \citep{Hoeffding1948} and the far reaching recent work of \citet{Pohle2025a}, we find central limit laws for the empirical coefficients.  Then we use the Gaussian limit distributions to develop interval estimators and tests in the usual ways.  

\subsection{Empirical versions}  \label{sec:empirical}

We begin the section by stating the plug-in empirical versions $\AKC_n$ and $\CID_n$ of the population quantities $\AKC(X,Y)$ at \eqref{eq:AKC} and $\CID(X,Y)$ at \eqref{eq:CID_AKC}, respectively.  

\begin{proposition}  \label{prop:AKC_empirical} 
The empirical asymmetric Kendall correlation equals
\begin{align}  \label{eq:AKC_n}
\AKC_n = \frac{\sum_{i=1}^n \sum_{j=1}^n \left( \one \{ X_i < X_j \} - \one \{ X_i > X_j \} \right) \one \{ Y_i < Y_j \}}{\sum_{i=1}^n \sum_{j=1}^n \one \{ Y_i < Y_j \}}, 
\end{align}
and the empirical C index equals
\begin{align}
\CID_n 
& = \frac{1}{2} \left( \AKC_n + 1 \right) \label{eq:CID_n} \\
& = \frac{\sum_{i=1}^n \sum_{j=1}^n s(X_i, X_j) \one \{ Y_i < Y_j \}}{\sum_{i=1}^n \sum_{j=1}^n \one \{ Y_i < Y_j \}}, \label{eq:CID2s_n}
\end{align}
where the function $s$ is defined at \eqref{eq:s}.
\end{proposition}

We emphasize that the expressions for $\AKC_n$ and $\CID_n$ at \eqref{eq:AKC_n} and \eqref{eq:CID2s_n} recover the provisions in the landmark paper by \citet[Section~5.5]{Harrell1996} that introduced the C index, provided the data are complete, with no censoring being present.\footnote{We note that \citet[p.~371]{Harrell1996} refer to $\AKC_n$ as the Somers' $D$ rank correlation index.  The handling of ties for these coefficients can be intricate, and subsequent work has differed in their treatment.  For detailed discussion, we point the reader at \citet[Section~4.1]{Sierra2025} and references therein.}  While naive implementations require $\cO(n^2)$ operations, the empirical version of Kendall's Tau along with $\AKC_n$ and $\CID_n$ can be computed in $\cO(n \log n)$ operations \citep{Knight1966}. 

Moving towards the plug-in empirical versions $\AGC_n$ and $\CMA_n$ of the population quantities $\AGC(X,Y)$ at \eqref{eq:AGC} and $\CMA(X,Y)$ at \eqref{eq:CMA}, respectively, we introduce notation for a range of equivalent representations.  Specifically, let $\bar{Q}_1, \ldots, \bar{Q}_n$ be the mid ranks \citep{Woodbury1940} associated with the covariate values $X_1, \ldots, X_n$, and let $\bar{R}_1, \ldots, \bar{R}_n$ be the mid ranks associated with the outcomes $Y_1, \ldots, Y_n$.\footnote{The mid rank $\bar{R}_j$ associated with $Y_j$ is $\bar{R}_j = \sum_{i=1}^n \one \{ Y_i < Y_j \} + \frac{1}{2} \sum_{i=1}^n \one \{ Y_i = Y_j \} + \frac{1}{2}$.}  Let $Z_1 < \cdots < Z_M$ denote the $M \in \{ 2, \ldots, n \}$ unique values of $Y_1, \ldots, Y_n$, and define
\begin{align}  \label{eq:N_c}
N_c = \sum_{i=1}^n \one \{ Y_i = Z_c \}, \quad c = 1, \ldots, M,
\end{align}
as the number of instances among the outcomes $Y_1, \ldots, Y_n$ that equal $Z_c$, so that $N_1 + \cdots + N_M = n$.  Then the data at \eqref{eq:data} can be written as
\begin{align}  \label{eq:data_alternative} 
(X_{11}, Z_1), \ldots, (X_{1N_1},Z_1), \, \ldots, \, (X_{M1}, Z_M), \ldots, (X_{M N_M},Z_M),
\end{align} 
where $X_{c1}, \ldots, X_{cN_c}$ are the covariate values that correspond to the unique outcome value $Z_c$.  Let $\widetilde{R}_1, \ldots, \widetilde{R}_M$ be the unique sorted mid ranks associated with $Z_1, \ldots, Z_M$, namely, $\widetilde{R}_1 = \frac{1}{2} (N_1 + 1)$ and $\widetilde{R}_c = \sum_{j=1}^{c-1} N_j + \frac{N_c+1}{2}$ for $c = 2, \ldots, M$.  Then the following representations hold. 

\begin{proposition}  \label{prop:AGC_empirical} 
The empirical asymmetric grade correlation equals
\begin{align}  \label{eq:AGC_n}
\AGC_n = \frac{\sum_{i=1}^n (\bar{Q}_i - \frac{1}{2}(n+1)) \hsp (\bar{R}_i - \frac{1}{2}(n+1))}{\sum_{i=1}^n (\bar{R}_i - \frac{1}{2}(n+1))^2}; 
\end{align}
and the empirical coefficient of monotone association equals
\begin{align}  
\CMA_n
& = \frac{1}{2} \left( \AGC_n + 1 \right) \label{eq:CMA_n} \\
& =  \frac{\sum_{k=1}^n \sum_{l=1}^n (\bar{R}_k - \bar{R}_l) \, \one\{Y_k < Y_l\} \, s(X_k, X_l)}{\sum_{k=1}^n \sum_{l=1}^n (\bar{R}_k - \bar{R}_l) \, \one\{Y_k < Y_l\}}  \label{eq:CMA2s_n} \\
& = \frac{\sum_{i=1}^{M-1} \sum_{j=i+1}^M (\widetilde{R}_j - \widetilde{R}_i) \sum_{k=1}^{N_i} \sum_{l=1}^{N_j} s(X_{ik}, X_{jl})}{\sum_{i=1}^{M-1} \sum_{j=i+1}^M (\widetilde{R}_j - \widetilde{R}_i) N_i N_j} \label{eq:CMA2s_alt_n}.
\end{align}
\end{proposition}

The representations \eqref{eq:AGC_n} and \eqref{eq:CMA_n} demonstrate that $\AGC_n$ and $\CMA_n$ can be computed in $\cO(n \log n)$ operations.  The equivalent expression at \eqref{eq:CMA2s_n} is the plug-in empirical version of the population level representation at \eqref{eq:CMA2s}.

The representation \eqref{eq:CMA2s_alt_n} is a variant of \eqref{eq:CMA2s_n} that facilitates comparison with the empirical coefficient of predictive ability ($\CPA_n$) proposed by \citet{Gneiting2022b}.  For data of the form \eqref{eq:data_alternative}, 
\begin{align}  \label{eq:CPA_n} 
\CPA_n = 
\frac{\sum_{i=1}^{M-1} \sum_{j=k+1}^M \sum_{k=1}^{N_i} \sum_{l=1}^{N_j} ( \hsp j-i) \, s(X_{ik},X_{jl})} 
     {\sum_{i=1}^{M-1} \sum_{j=i+1}^M ( \hsp j-i) \, N_i N_j}.  
\end{align}
However, \citet{Gneiting2022b} do not provide a population version of $\CPA_n$, nor do we believe that such can be developed, except for settings where $\CPA_n$ equals $\CMA_n$.

In general, $\CMA_n$ and $\CPA_n$ are distinct quantities, but they equal each other in important special cases.  Specifically, when the outcome is dichotomous, then $M = 2$ and
\begin{align}  \label{eq:AUC_n} 
\CMA_n = \CPA_n = \AUC_n = \frac{\sum_{k=1}^{N_1} \sum_{l=1}^{N_2} s(X_{1k},X_{2l})}{N_1 N_2},
\end{align}
where $\AUC_n$ is the empirical version of the $\AUC$ criterion at \eqref{eq:AUC}.  Furthermore, when the outcomes $Y_1, \ldots, Y_n$ at \eqref{eq:data} are pairwise distinct, then $M = n$ and $N_c = 1$ for $c = 1, \ldots, n$ at \eqref{eq:N_c}, and the following holds.

\begin{corollary}  \label{cor:CPA} 
Suppose that the outcomes\/ $Y_1, \ldots, Y_n$ for the data at\/ \eqref{eq:data} are pairwise distinct with order statistics $Z_1 < \cdots < Z_n$.  For $i = 1, \ldots, n - 1$, let 
\begin{align*}
\AUC_n^{\hsp (i/n)} = \frac{\sum_{k=1}^n \sum_{l=1}^n s(X_k,X_l) \hsp \one \{ Y_k \leq Z_i < Y_l \}}{i \hsp (n-i)}
\end{align*}
denote $\AUC_n$ for the induced data\/ $(X_1, \one \{ Y_1 > Z_i \}), \ldots, (X_n, \one \{ Y_n > Z_i \})$, where the outcome is dichotomized at the threshold $Z_i$.  Then
\begin{align}   \label{eq:CMA2AUC_n}
\CMA_n = \CPA_n = \frac{6}{n} \sum_{i=1}^{n-1} \frac{i \hsp (n-i)}{(n-1) \hsp (n+1)} \AUC_n^{\hsp (i/n)}. 
\end{align}
\end{corollary}

We interpret the representation \eqref{eq:CMA2AUC_n} for $\CMA_n$ as an empirical version of the relationship \eqref{eq:CMA2AUC} for $\CMA$ at the population level.  Essentially the same representation for $\CMA_n$ holds in slightly more general settings, namely, when all class sizes $N_1 = \cdots = N_M$ at \eqref{eq:N_c} are equal.

\subsection{Asymptotic distributions}  \label{sec:asymptotic}

We turn to the asymptotic distribution of $\AKC_n$ and $\CID_n$, and $\AGC_n$ and $\CMA_n$, with the development drawing heavily on theory recently developed by \citet{Pohle2025a}.  Evidently, it suffices to consider $\AKC_n$ and $\AGC_n$, and we abbreviate the respective population coefficients as $\AKC = \AKC(X,Y)$ and $\AGC = \AGC(X,Y)$, respectively.  For $\CID_n = \frac{1}{2} \hsp (\AKC_n + 1)$ and $\CMA_n = \frac{1}{2} \hsp (\AGC_n + 1)$ the asymptotic variances in the subsequent need to be scaled by a factor of $\frac{1}{4}$.

We begin by stating central limit theorems in the case of independent, identically distributed sequences.  For clarity, we formalize the setting.

\begin{scenario}  \label{sce:1}
Consider a sequence $(X_i,Y_i)_{i = 1, 2, \ldots}$ of independent random vectors, where each $(X_i,Y_i)$ has the same law $\myP$ as $(X,Y)$, with $X$ and $Y$ having strictly positive variance.
\end{scenario}

As before, let $F(x) = \myP(X \leq x)$ and $\Fbar(x) = \myP(X < x) + \frac{1}{2} \hsp \myP(X = x)$ denote the cumulative distribution function (CDF) and the mid distribution function (MDF) of $X$.  Let $G$ and $\Gbar$ be the CDF and MDF of $Y$, and let  
\begin{align*}
\Hbar(x,y) = \myP( X < x, Y < y) + \frac{1}{2} \, \myP( X = x, Y < y) + \frac{1}{2} \, \myP( X  < x, Y = y) + \frac{1}{4} \, \myP( X = x, Y = y)
\end{align*}
denote the bivariate MDF of $(X,Y)$, where $x, y \in \real$.  Let $\tau = \tau(X,Y)$ from \eqref{eq:tau}, define 
\begin{align}  \label{eq:K_Kendall}
K^{(\operatorname{Kendall})}(x,y) = 4 \, \Hbar(x, y) - 2 \left( \Fbar(x) + \Gbar(y) \right) + 1 - \tau
\end{align}
for $x, y \in \real$, and let  
\begin{align}  \label{eq:K_2}
K^{(2)}(y)  =  G(y) - G(y-) - \twogran_G
\end{align}
for $y \in \real$, where $G(y-) = \lim_{z \uparrow y} G(z)$ and $\twogran_G = \twogran(Y)$.  Finally, let $\threegran_F = \threegran(X)$ and $\threegran_G = \threegran(Y)$.  We are now ready to state the limit distribution of $\AKC_n$, where we express the asymptotic variance in terms of the kernels $K^{(\operatorname{Kendall})}$ and $K^{(2)}$ at \eqref{eq:K_Kendall} and \eqref{eq:K_2}, respectively. 

\begin{theorem}  \label{thm:AKC_CLT}
Under Scenario \ref{sce:1} it holds that 
\begin{align}  \label{eq:AKC_CLT}
\sqrt{n} \left( \AKC_n - \AKC \right) \stackrel{d}{\longrightarrow} \: \cN \left( 0, \sigma^2 \right) ,
\end{align}
where
\begin{align}  \label{eq:AKC_sigma}
\sigma^2 = \frac{4}{(1 - \twogran_G)^2} \; \myE \hspm \left( K^{(\operatorname{Kendall})}(X, Y) + \frac{\tau}{1 - \twogran_G} K^{(2)}(Y) \right)^{\! 2}
\end{align}
In particular, if\/ $X$ and\/ $Y$ are independent then
\begin{align}  \label{eq:AKC_CLT_independent}
\sqrt{n} \hsp \AKC_n \stackrel{d}{\longrightarrow} \: \cN \left( 0, \hsp \frac{4}{9} \frac{(1 - \threegran_F) (1 - \threegran_G)}{(1 - \twogran_G)^2}  \right) \! .
\end{align}
\end{theorem}

The theorem subsumes a wide range of classical results, as we discuss below, once we have stated the respective result for $\AGC_n$.  To this end, define $\Ftilde(x) = \myE \hsp [\Hbar(x,Y)]$ for $x \in \real$ and $\Gtilde(y) = \myE \hsp [\Hbar(X,y)]$ for $y \in \real$, let $\rhoG = \rhoG(X,Y)$ from \eqref{eq:rhoG}, and let 
\begin{align}  \label{eq:K_Spearman}
K^{(\operatorname{Spearman})}(x,y) = 4 \left( \Ftilde(x) + \Gtilde(y) + \Fbar(x) \hsp \Gbar(y) - \Fbar(x) - \Gbar(y) \right) + 1 - \rhoG 
\end{align}
for $x, y \in \real$.  Furthermore, let 
\begin{align}  \label{eq:K_3}
K^{(3)}(y)  = (G(y) - G(y-))^2 - \threegran_G
\end{align}
for $y \in \real$.  We now state the limit distribution of $\AGC_n$ in terms of the kernels $K^{(\operatorname{Spearman})}$ and $K^{(3)}$ at \eqref{eq:K_Spearman} and \eqref{eq:K_3}, respectively. 

\begin{theorem}  \label{thm:AGC_CLT}
Under Scenario \ref{sce:1} it holds that 
\begin{align}  \label{eq:AGC_CLT}
\sqrt{n} \left( \AGC_n - \AGC \right) \stackrel{d}{\longrightarrow} \: \cN \left( 0, \sigma^2 \right) \! ,
\end{align}
where
\begin{align}  \label{eq:AGC_sigma}
\sigma^2 = \frac{9}{(1 - \threegran_G)^2} \; \myE \hspm \left( K^{(\operatorname{Spearman})}(X,Y) + \frac{\rhoG}{1 - \threegran_G} K^{(3)}(Y) \right)^2.
\end{align}
In particular, if\/ $X$ and\/ $Y$ are independent then
\begin{align}  \label{eq:AGC_CLT_independent}  
\sqrt{n} \hsp \AGC_n \stackrel{d}{\longrightarrow} \: \cN \left( 0, \frac{1 - \threegran_F}{1 - \threegran_G} \right) \! .
\end{align}
\end{theorem}

Theorems \ref{thm:AKC_CLT} and \ref{thm:AGC_CLT} generalize, unify, and subsume a very wide range of classical results.  For a dichotomous outcome $Y = Y_d$ with success probability $\pi \in (0,1)$, it holds that 
\begin{align*} 
\AKC(X,Y_d) = \AGC(X,Y_d) = 2 \AUC(X,Y_d) - 1.  
\end{align*} 
In this setting, the limit laws at \eqref{eq:AKC_CLT} and \eqref{eq:AGC_CLT} equal each other and simplify to classical results of \citet{Bamber1975} and \citet{DeLong1988}, which we state in terms of $\AUC = \AUC(X,Y_d)$ at \eqref{eq:AUC} and the empirical version $\AUC_n$ at \eqref{eq:AUC_n}.  The literature traditionally operates in terms of the conditional MDFs 
\begin{align*}
\textstyle
\Fbar_0(x) = \myP( X < x \mid Y_d = 0) + \frac{1}{2} \myP( X = x \mid Y_d = 0),     
\quad 
\Fbar_1(x) = \myP( X < x \mid Y_d = 1) + \frac{1}{2} \myP( X = x \mid Y_d = 1),  
\end{align*}
respectively, so that $\Fbar(x) = (1 - \pi) \Fbar_0(x) + \pi \Fbar_1(x)$ for $x \in \real$.  With 
\begin{align*}
\sigma_0^2 = \var( \Fbar_1(X) \mid Y_d = 0), \quad \sigma_1^2 = \var( \Fbar_0(X) \mid Y_d = 1), 
\end{align*}
tedious but straightforward calculations demonstrate\footnote{Details are available from the authors upon request.} that the relations at \eqref{eq:AKC_CLT} and \eqref{eq:AGC_CLT} equal each other and translate into the classical result, namely, 
\begin{align}  \label{eq:AUC_CLT}  
\sqrt{n} \left( \AUC_n - \AUC \right) \stackrel{d}{\longrightarrow} \: \cN \left( 0, \hsp \frac{\sigma_1^2}{\pi} + \frac{\sigma_0^2}{1 - \pi} \right) .
\end{align}
In particular, if $X$ and $Y_d$ are independent then $\AUC = \frac{1}{2}$ and $\sigma_0^2 = \sigma_1^2 = \var(\Fbar(X)) = \frac{1}{12} (1 - \threegran_F)$, so that 
\begin{align*}
\sqrt{n} \left( \AUC_n - \frac{1}{2} \right) \stackrel{d}{\longrightarrow} \: \cN \left( 0, \hsp \frac{1 - \threegran_F}{12 \, \pi \hsp (1-\pi)} \right) .
\end{align*}

Next we discuss the special case in which $X = X_c$ and $Y = Y_c$ are continuous.  Then $\AKC$ equals Kendall's Tau at \eqref{eq:tau} and the asymptotic variance at \eqref{eq:AKC_sigma} specializes to 
\begin{align*}
\sigma^2 = 4 \; \myE \hspm \left( K^{(\operatorname{Kendall})}(X_c,Y_c)^2 \right) ,    
\end{align*}
reducing further to the constant $\sigma^2 = \frac{4}{9}$ in the case of independence, and recovering special cases of results in \citet{Kendall1938}, \citet{Hoeffding1948}, \citet{Dengler2010}, and \citet{Pohle2025a}, among others.  Likewise, $\AGC$ equals Spearman's rank correlation coefficient $\rhoS$ and grade correlation $\rhoG$ and the asymptotic variance at \eqref{eq:AGC_sigma} specializes to 
\begin{align*}
\sigma^2 = 9 \; \myE \hspm \left( K^{(\operatorname{Spearman})}(X_c,Y_c)^2 \right) ,    
\end{align*}
reducing further to $\sigma^2 = 1$ in the case of independence,\footnote{We recall that $\sigma^2 = \frac{4}{9}$ for Kendall's Tau in the case of independence, which reflects the fact that, for many continuous distributions with weak dependencies, the ratio of Kendall's Tau over Spearman's Rho equals about two over three \citep{Fredricks2007}.} and overlapping with extant results of \citet{Hoeffding1948}, \citet{Cifarelli1996}, \citet{Borkowf2002}, \citet{Genest2013}, and \citet{Pohle2025a}, among others.

In the general case of a nondegenerate joint distribution $\myP$ of $(X,Y)$, the limit variance in Theorem \ref{thm:AKC_CLT} corresponds to finite sample expressions in \citet{Daniels1947}, \citet{Cliff1991}, \citet{Dengler2010}, \citet{Obuchowski2006}, \citet{Demler2017}, and \citet{Zou2022}, and asymptotic arguments in \citet{Kang2015}.  Theorem \ref{thm:AGC_CLT} is original in the general case.

Let us now turn to the setting of $m$ covariates, features, or markers $X^{(1)}, \ldots, X^{(m)}$, which we interpret as competing or complementary predictors for the real-valued outcome $Y$.

\begin{scenario}  \label{sce:2}
Consider a sequence 
\begin{align}  \label{eq:data_MV}
\left( X_i^{(1)}, \ldots, X_i^{(m)}, Y_i \right)_{i = 1, 2, \ldots}
\end{align}
of independent random vectors, where each $(X_i^{(1)}, \ldots, X_i^{(m)}, Y_i)$ has the same distribution $\myP$ as $(X^{(1)}, \ldots, X^{(m)}, Y)$, with $X^{(1)}, \ldots, X^{(m)}$, and $Y$ all having strictly positive variance.
\end{scenario}

For $k = 1, \ldots, m$, we let 
\begin{align}  \label{eq:tau_k}
\tau^{(k)} = \myE \left( \sgn \! \left( X^{(k)}  - {X^{(k)}}' \, \right) \sgn \hsp ( Y - Y' \hsp ) \right)
\end{align}
and $\AKC^{(k)} = \AKC(X^{(k)},Y)$, and we write $\AKC^{(k)}_n$ for the asymmetric Kendall correlation under the empirical law of the first $n$ tuples in the sequence at \eqref{eq:data_MV}.  Furthermore, we let $\Fbar^{(k)}$ denote the univariate MDF of $X^{(k)}$, and we let $\Hbar^{(k)}$ denote the bivariate MDF of $(X^{(k)},Y)$.  In the subsequent result, we express the entries of the asymptotic covariance matrix in terms of the kernel functions 
\begin{align*}
K^{(\operatorname{Kendall}, \, k)}(x,y) = 4 \hsp \Hbar^{(k)}(x, y) - 2 \left( \Fbar^{(k)}(x) + \Gbar(y) \right) + 1 - \tau^{(k)}
\end{align*} 
for $k = 1, \ldots, m$, and $K^{(2)}(y)$ at \eqref{eq:K_2}, respectively.

\begin{theorem}  \label{thm:AKC_CLT_MV}
Under Scenario \ref{sce:2} it holds that
\begin{align*}
\sqrt{n} \left( \begin{pmatrix} \AKC_n^{(1)} \\ \vdots \\ \AKC_n^{(m)} \end{pmatrix} - \begin{pmatrix} \AKC^{(1)} \\ \vdots \\ \AKC^{(m)} \end{pmatrix} \right) 
\stackrel{d}{\longrightarrow} \: \cN \! \left( \begin{pmatrix} 0 \\ \vdots \\ 0 \end{pmatrix} , \Sigma = \left( \Sigma^{(k,l)} \right)_{k,l=1}^m \right) ,
\end{align*}
where 
\begin{align}  
\Sigma^{(k,l)} & = \frac{4}{(1 - \twogran_G)^2} \: \times \label{eq:AKC_sigma_kl} \\
& \myE \left[ \left( K^{(\operatorname{Kendall}, \, k)}(X^{(k)}, Y) + \frac{\tau^{(k)}}{1 - \twogran_G} \, K^{(2)}(Y) \right) \! \left( K^{(\operatorname{Kendall}, \, l)}(X^{(l)}, Y) + \frac{\tau^{(l)}}{1 - \twogran_G} \, K^{(2)}(Y) \right) \nonumber
\right]
\end{align}
for\/ $k, l = 1, \ldots, m$.
\end{theorem}

For $k = 1, \ldots, m$, we now let 
\begin{align*}
\rhoG^{\, (k)} = \rhoG(X^{(k)},Y) = 12 \hsp \cov(\Fbar^{(k)}(X^{(k)}),\Gbar(Y))
\end{align*}
and $\AGC^{(k)} = \AGC(X^{(k)},Y)$, and we write $\AGC^{(k)}_n$ for the asymmetric grade correlation under the empirical law of the first $n$ tuples in the sequence at \eqref{eq:data_MV}.  Furthermore, we define the nondecreasing functions $\Ftilde^{(k)}(x) = \myE \hsp [\Hbar^{(k)}(x,Y)]$ for $x \in \real$ and $\Gtilde^{(k)}(y) = \myE \hsp [\Hbar^{(k)}(X^{(k)}, y)]$ for $y \in \real$.  In the subsequent theorem we express the asymptotic covariance matrix in terms of the kernel functions 
\begin{align*}
K^{(\operatorname{Spearman}, k)}(x,y) = 4 \left( \Ftilde^{(k)}(x) + \Gtilde^{(k)}(y) + \Fbar^{(k)}(x) \hsp \Gbar(y) - \Fbar^{(k)}(x) - \Gbar(y) \right) + 1 - \rhoG^{\, (k)}
\end{align*} 
for $k = 1, \ldots, m$, and $K^{(3)}(y)$ at \eqref{eq:K_3}, respectively.

\begin{theorem}  \label{thm:AGC_CLT_MV}
Under Scenario \ref{sce:2} it holds that
\begin{align*}
\sqrt{n} \left(
\begin{pmatrix} \AGC_n^{(1)} \\ \vdots \\ \AGC_n^{(m)} \end{pmatrix} - \begin{pmatrix} \AGC^{(1)} \\ \vdots \\ \AGC^{(m)} \end{pmatrix} 
\right) \stackrel{d}{\longrightarrow} \: \cN \! \left( \begin{pmatrix} 0 \\ \vdots \\ 0 \end{pmatrix} , \Sigma = \left( \Sigma^{(k,l)} \right)_{k,l=1}^m
\right) ,
\end{align*}
where 
\begin{align}
\Sigma^{(k,l)} & = \frac{9}{(1 - \threegran_G)^2} \: \times \label{eq:AGC_sigma_kl} \\
& \myE \left[ \left( K^{(\operatorname{Spearman}, \, k)}(X^{(k)}, Y) + \frac{\rhoG^{\, (k)}}{1 - \threegran_G} \, K^{(3)}(Y) \right) \! \left( K^{(\operatorname{Spearman}, \, l)}(X^{(l)}, Y) + \frac{\rhoG^{\, (l)}}{1 - \threegran_G} \, K^{(3)}(Y) \right) \right] \nonumber
\end{align}
for\/ $k, l = 1, \ldots, m$.
\end{theorem}

As in the univariate case, Theorems \ref{thm:AKC_CLT_MV} and \ref{thm:AGC_CLT_MV} nest extant findings when the outcome $Y$ is dichotomous \citep{DeLong1988} and when $X$ and $Y$ are continuous \citep{Hoeffding1948, Gaisser2010, Pohle2025a}.

If we wish to use these asymptotic results to generate confidence intervals or test hypotheses, we need to replace the population quantities that appear in the entries of the asymptotic covariance matrices by suitable estimates.  For doing this, we follow \citet{Pohle2025a} and proceed stepwise, by evaluating the population quantities under the empirical law of the first $n$ terms in \eqref{eq:data_MV}.  

In the case of $\AKC$ we find multi-step plug-in estimates as follows.  We use subscripts to distinguish sample quantities from population quantities, and to avoid double subscripts we write $\twogran_n$ for the sample version of the 2-granularity $\twogran_G$. 
\begin{enumerate}
\item Compute the quantities $\tau^{(k)}_n$, $\twogran_n$, $G_n(Y_i) - G_n(Y_i-)$, $\Gbar_n(Y_i)$, $\Fbar^{(k)}_n(X_i^{(k)})$, and $\Hbar^{(k)}_n(X_i^{(k)},Y_i)$, where $i = 1, \dots, n$ and $k = 1, \ldots, m$.  
\item Employ the first step estimates to compute the kernel values $K^{(\operatorname{Kendall}, \, k)}_n(X_i^{(k)},Y_i)$ and $K^{(2)}_n(Y_i)$, where $i = 1, \ldots, n$ and $k = 1, \ldots, m$.
\item Plug the first and second step estimates into the sample version of \eqref{eq:AKC_sigma_kl}.
\end{enumerate}
We denote the resulting multi-step plug-in estimate by $\hat\Sigma_n$ and refer to its entries as $\hat\Sigma^{(k,l)}_n$, where $k, l = 1, \ldots, m$.  When $m = 1$ we denote the estimate of the asymptotic variance by $\hat\sigma^2_n$.  Under Scenario \ref{sce:2} the arguments in \citet{Pohle2025a} imply that $\hat\Sigma_n \rightarrow \Sigma$ in probability as $n \rightarrow \infty$.  The computational complexity of the naive implementation of the variance estimator is $\cO(m \hsp n^2 + m^2 \hsp n)$; using Fenwick trees \citep{Fenwick1994} the effort reduces to $\cO(m \hsp n \log n + m^2 \hsp n)$.

In the case of $\AGC$, we find multi-step plug-in estimates as follows.  Again we use subscripts to distinguish sample quantities from population quantities, and to avoid double subscripts we write $\threegran_n$ for the sample version of the 3-granularity $\threegran_G$. 
\begin{enumerate} 
\item Compute the quantities $\threegran_n$, $G_n(Y_i) - G_n(Y_i-)$, $\Gbar_n(Y_i)$, $\Fbar^{(k)}_n(X_i^{(k)})$, and $\Hbar^{(k)}_n(X_i^{(k)},Y_i)$, where $i = 1, \dots, n$ and $k = 1, \ldots, m$.  
\item Use the first step estimates to find $\rhoG_n^{\, (k)}$, $\Ftilde^{(k)}_n(X_i^{(k)})$, and $\Gtilde^{(k)}_n(Y_i)$, where $i = 1, \ldots, n$ and $k = 1, \ldots, m$.
\item Employ the first and second step estimates to compute the kernel values $K^{(\operatorname{Spearman}, \, k)}_n(X_i^{(k)},Y_i)$ and $K^{(3)}_n(Y_i)$, where $i = 1, \ldots, n$ and $k = 1, \ldots, m$.
\item Plug the first, second, and third step estimates into the sample version of \eqref{eq:AGC_sigma_kl}.  
\end{enumerate}
Again, we denote the resulting multi-step plug-in estimates by $\hat\Sigma^{(k,l)}_n$ and $\hat\sigma^2_n = \hat\Sigma^{(1,1)}_n$, respectively, and under Scenario \ref{sce:2} it holds that $\hat\Sigma_n \rightarrow \Sigma$ in probability as $n \rightarrow \infty$ \citep{Pohle2025a}. 

Let us now discuss computational facets.  The granularities can be computed from the empirical frequency table of the outcomes in $\cO(n)$ operations.  Consequently, the computational bottlenecks lie in the computation of the bivariate rank quantities and in the construction of the $m \times m$ matrix $\hat\Sigma_n$ in the final step, which require (naively implemented) up to $\cO(m \hsp n^2 + m^2n)$ operations.  Using prefix-sum techniques \citep{Fenwick1994}, the cost per predictor reduces from $\cO(n^2)$ to $\cO(n \log n)$, for a total cost of $\cO(m \hsp n \log n + m^2n)$ operations.  

Alternative, asymptotically equivalent variance estimators based on the infinitesimal jackknife \citep{Efron1982} share this favorable rate of computational complexity and, in the case of $\CID$, have been implemented in the \texttt{survival} package \citep{survival} in \textsf{R} \citep{R}.  We refrain from a technical discussion and refer to \citet{Arvesen1969}, \citet{Jaeckel1972}, \citet{Schechtman1991}, \citet{Antolini2004}, \citet{Newson2006}, and \citet{survival} for details.  For a dichotomous outcome, the abundance of ties can be leveraged, and we point at \citet{DeLong1988}, \citet{Sun2014}, and \citet{Demler2017} for details on variance estimators for $\AUC$ and differences in $\AUC$.  In the continuous case, \citet{Gaisser2010} propose a nonparametric bootstrap technique for variance estimation.

Finally, we note from \citet{Pohle2025a} that more general versions of Propositions \ref{prop:AKC_MV} and \ref{prop:AGC_MV} in Appendix \ref{app:proofs_sample} hold in dependent time series settings.  Therefore, our central limit theorems can be generalized as well.  In practice, it suffices to replace the multi-step plug-in estimates of the population quantities at \eqref{eq:AKC_sigma_kl} and \eqref{eq:AGC_sigma_kl} by heteroskedasticity and autocorrelation consistent estimates \citep{Newey1987}, for which \citet{Pohle2025a} provide implementation details and prove consistency.

Thus far we have operated under Scenario \ref{sce:2} with interest in inference for the relations between $m$ predictors and a shared outcome $Y$.  More generally, we might consider a sequence
\begin{align}  \label{eq:data_MV_generalized}
\left( X_i^{(1)}, \ldots, X_i^{(m)} \right)_{i = 1, 2, \ldots}
\end{align}
of independent random vectors, where each $(X_i^{(1)}, \ldots, X_i^{(m)})$ has the same distribution $\myP$ as $(X^{(1)}, \ldots, X^{(m)})$, but now with interest in inference for the $m^2$ quantities 
\begin{align*}
\AKC^{(i,j)} = \AKC(X^{(i)}, X^{(j)}) \quad \text{or} \quad \AGC^{(i,j)} = \AGC(X^{(i)}, X^{(j)}), 
\end{align*}
respectively, where $i, j = 1, \ldots, m$.  The resulting multivariate normal limit distribution is degenerate\footnote{Degeneracies stem not only from the inclusion of $\AKC^{(i,i)}$ or $\AGC^{(i,i)}$, but may also arise from indices satisfying $k = j$ and $l = i$, or from issues discussed by \citet[Section~2.4]{Dengler2010}.} and has a singular covariance matrix $\Sigma$ of dimension $m^2 \times m^2$ with components $\Sigma_{(i,j),(k,l)}$ for $i, j, k, l = 1, \ldots, m$.  The details are tedious and we leave them to future work.  

\subsection{General tests of DeLong type}  \label{sec:DeLong}

With variance estimates at hand, we can use the limit distributions from Theorems \ref{thm:AKC_CLT}--\ref{thm:AGC_CLT} and \ref{thm:AKC_CLT_MV}--\ref{thm:AGC_CLT_MV} to generate confidence intervals and test hypotheses about $\AKC$, $\AGC$, $\CID$, and $\CMA$.  As these measures nest the $\AUC$ measure for dichotomous outcomes, and Kendall's Tau and Spearman's Rho for continuous variables, inference tools for these classical measures arise as special cases.  Evidently, the methods are the same for all measures, and we focus the discussion on the $\AGC$ measure.  

To generate confidence intervals, the asymptotic variances at \eqref{eq:AKC_sigma} and \eqref{eq:AGC_sigma} can be employed as usual.  For testing, let us address the most basic case first, namely, tests of independence.  Under independence of $X$ and $Y$, it holds that $\AGC(X,Y) = 0$, and we use the test statistic 
\begin{align}  \label{eq:simple}
T_n = \sqrt{n} \; \frac{\AGC_n}{\hat\sigma_n},
\end{align}
where $\AGC_n$ and $\hat\sigma_n$ are the estimates from Sections \ref{sec:empirical} and \ref{sec:asymptotic}, respectively.  By Theorem \ref{thm:AGC_CLT} and the consistency of $\hat\sigma_n$, $T_n$ is asymptotically standard normal under independence.  With $\Phi$ denoting the CDF of a standard normal distribution, we find a one-sided $p$-value, $p = \Phi(T_n)$, or a two-sided $p$-value, $p = 2 \left( 1 - \Phi(|T_n| \right)$, in the usual way.  Similarly, we test simple hypotheses of the form $\AGC = a_0 \in (-1,1)$ or $\CMA = b_0 \in (0,1)$.

Next we discuss tests of the equality of two or more $\AGC$ or $\CMA$ values for a shared outcome, as developed by \citet{DeLong1988} in the dichotomous case.  Tests of this type address the ubiquitous task of the comparison of the potential predictive ability of competing covariates, scores, features, or markers, as frequently encountered in machine learning research \citep{Rainio2024}.  Specifically, we consider the case $m = 2$ in Scenario \ref{sce:2}, where we interpret $X^{(1)}$ and $X^{(2)}$ as predictors of the outcome $Y$ and test the hypothesis 
\begin{align}  \label{eq:H0_pairwise}
H_0 : \AGC(X^{(1)},Y) = \AGC(X^{(2)},Y).
\end{align}
For this we use the test statistic 
\begin{align}  \label{eq:pairwise}
\Delta_n = \sqrt{n} \; \frac{\AGC^{(1)}_n - \AGC^{(2)}_n}{\left( \hat\Sigma_n^{(1,1)} + \hat\Sigma_n^{(2,2)} - 2 \hsp \hat\Sigma_n^{(1,2)} \right)^{1/2}},
\end{align}
where $\hat\Sigma_n$ is the multi-step plug-in estimate from Section~\ref{sec:asymptotic}.  By Theorem \ref{thm:AGC_CLT_MV} and the consistency of $\hat\Sigma_n$, $\Delta_n$ is asymptotically standard normal under \eqref{eq:H0_pairwise}.  Therefore, one-sided and two-sided tests yield $p$-values $1 - \Phi(\Delta_n)$ and $2 \left( 1 - \Phi(|\Delta_n|) \right)$, respectively.  Evidently, tests for the equality of $\AUC$ (for dichotomous outcomes) and Spearman's Rho (in the continuous case) arise as special cases.  Similarly, tests based on $\AKC$ nest tests for the equality of $\AUC$ (for dichotomous outcomes) and Kendall's Tau (in the continuous case).  For dichotomous outcomes, the tests based on $\AGC$ and $\AKC$ coincide and reduce to the classical test of \citet{DeLong1988}, except for a minor difference in the finite-sample variance estimates.\footnote{For the data at \eqref{eq:data}, let $n_0 = \sum_{i=1}^n \one (Y_i = 0)$ and $n_1 = \sum_{i=1}^n \one (Y_i = 1)$.  The two estimators differ only in the divisor for $\sigma_0^2$ and $\sigma_1^2$: We use $n_0$ and $n_1$, whereas \citet{DeLong1988} use $n_0 - 1$ and $n_1 - 1$, with all else being equal.}  In this light, we view the two tests --- based on $\AGC$ and $\AKC$, respectively --- as generalized tests of \citet{DeLong1988} type that apply to all real-valued outcomes, covering the full range from dichotomous to continuous variables. 

In the case of $m \geq 2$ predictors $X^{(1)}, \ldots, X^{(m)}$ for a shared outcome $Y$ we adopt notation and terminology from Theorem \ref{thm:AGC_CLT_MV} and test the hypothesis that $\AGC^{(1)} = \cdots = \AGC^{(m)}$ as follows.  Let $L$ be a fixed matrix with $m$ columns, with each row representing a contrast.  For example, when $m = 4$ the matrix
\begin{align*}
L = \begin{bmatrix}
1 & -1 & 0 & 0 \\
1 & 0 & -1 & 0 \\
1 & 0 & 0 & -1 \\
0 & 1 & -1 & 0 \\
0 & 1 & 0 & -1 \\
0 & 0 & 1 & -1 \\ \end{bmatrix}
\end{align*}
represents all six possible contrasts.  Let $\AGC_n = (\AGC^{(1)}_n, \ldots, \AGC^{(m)}_n)$.  Under the null hypothesis, Theorem \ref{thm:AGC_CLT_MV} and the consistency of the estimator $\hat\Sigma_n$ for $\Sigma$ imply that the test statistic 
\begin{align*}
\chi_n = \AGC_n L' \hsp \left( L \hsp \hat\Sigma_n \hsp L' \right)^{-1} L \hsp \AGC_n'
\end{align*}
is asymptotically chi-square distributed with degrees of freedom equal to the rank $\ell$ of the matrix $L \Sigma L'$.  Accordingly, $p$-values can be found from the CDF of the chi-square distribution with $\ell$ degrees of freedom.  This test and its analogue for $\AKC$ also generalize the \citet{DeLong1988} methodology that applies in the nested case when the outcome is dichotomous.

\subsection{Simulation examples}  \label{sec:simulation} 

In this section, we illustrate our methods in simple simulation settings.  In the first example, we address the continuous case, where $\AGC$ and $\AKC$ are symmetric and equal Spearman's Rho and Kendalls Tau, respectively.  Specifically, we let $X^{(0)}, Z^{(1)}$, and $Z^{(2)}$ be independent and standard normal.  Conditionally on $X^{(0)}$, we let the outcome $Y$ be normally distributed with mean $X^{(0)}$ and variance 1, and we seek to compare the correlated predictors $X^{(1)} = X^{(0)} + Z^{(1)}$ and $X^{(2)} = X^{(0)} + Z^{(2)}$.  By construction, $(X^{(1)},Y)$ and $(X^{(2)},Y)$ are equal in distribution, so the null hypotheses 
\begin{align*}
H_0 : \AGC(X^{(1)},Y) = \AGC(X^{(2)},Y) \quad \textrm{and} \quad H_0 : \AKC(X^{(1)},Y) = \AKC(X^{(2)},Y) 
\end{align*}
of equal correlation values are satisfied.  In a second example, we retain the above setting, but discretize $X^{(1)}$, $X^{(2)}$, and $Y$, by rounding to the nearest integer.  Evidently, the above null hypotheses remain valid.  We recall that a valid test generates $p$-values with a  uniform distribution on the unit interval under the null hypothesis \citep[Lemma 3.3.1]{Lehmann2022}, which guarantees that tests at all levels attains the nominal size.

\begin{figure}[tb]  

\centering

\includegraphics[width = 0.95 \textwidth]{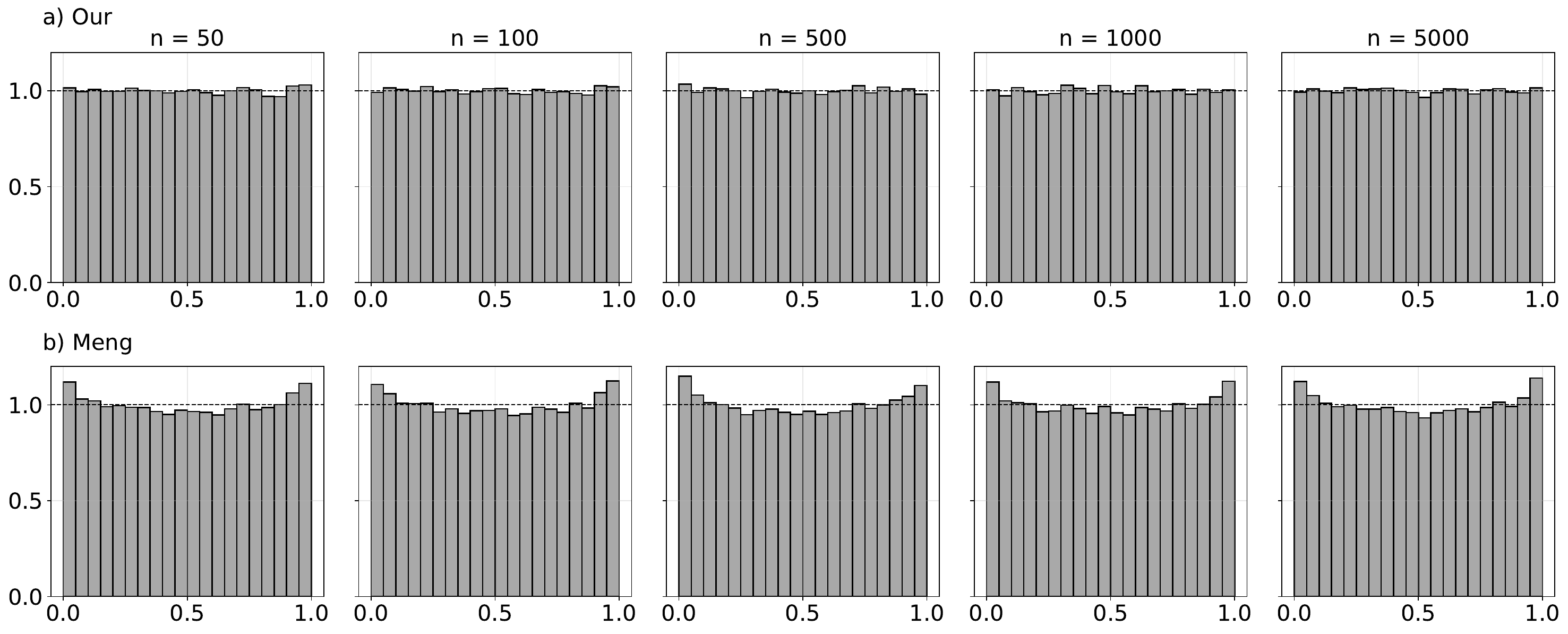}
\caption{Histograms of one-sided $p$-values in tests of equal $\AGC$ in the first, continuous simulation example, where $\AGC$ equals Spearman's Rho, using a) the general test of DeLong type with test statistic at \eqref{eq:pairwise}, and b) the method of \citet{Myers2014} with the \citet{Meng1992} adjustment for correlation.  The sample size $n$ equals $50$, $100$, $500$, $1000$, and $5000$ (left to right), and the number of Monte Carlo replicates is $100,000$.  \label{fig:sim_1}}

\bigskip

\includegraphics[width = 0.95 \textwidth]{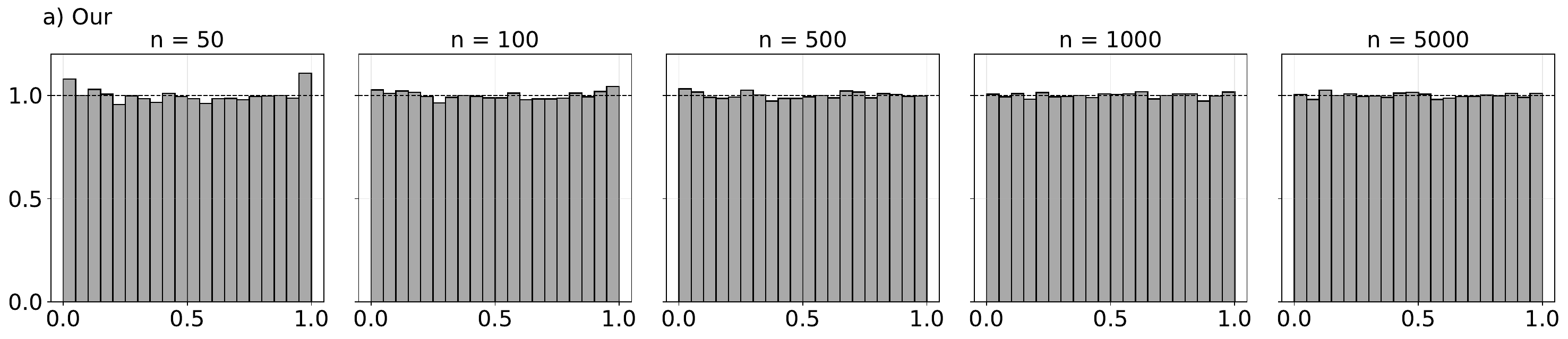}
\caption{Histograms of one-sided $p$-values in tests of equal $\AGC$ in the second, discrete simulation example, using the general test of DeLong type with test statistic at \eqref{eq:pairwise}.  The sample size $n$ equals $50$, $100$, $500$, $1000$, and $5000$ (left to right), and the number of Monte Carlo replicates is $100,000$.  \label{fig:sim_2a}}

\end{figure}

In the first, continuous simulation setting $\AGC$ equals Spearman's Rho at \eqref{eq:rhoS}.  Perhaps surprisingly, tests targeted specifically at the comparison of two Spearman rank correlation coefficients do not appear to be available.  Instead, the extant literature recommends that Spearman's Rho be treated as if it was a Pearson correlation, and then standard methods for the comparison of Pearson correlation coefficients be used \citep{Myers2014}.  Therefore, we compare our test to the test for the equality of correlated Pearson correlation coefficients developed by \citet{Meng1992},\footnote{Specifically, we apply Fisher's transformation \citep{Fisher1915} to $\AGC_n^{(1)}$ and $\AGC_n^{(2)}$, to obtain $Z^{(1)}$ and $Z^{(2)}$, respectively, and then use eq.~(1) of \citet{Meng1992} to adjust for correlation.} as implemented in the \texttt{cocor} package \citep{Diedenhofen2015} for \textsf{R}.  Figure \ref{fig:sim_1} shows histograms of one-sided $p$-values for tests at sample size $n = 50$, $100$, $500$, $1000$, and $5000$, respectively.  The test for equal Spearman's Rho with test statistic at \eqref{eq:pairwise} yields essentially uniform histograms even when $n = 50$.  In contrast, the method of \citet{Myers2014} with the \citet{Meng1992} correction yields slightly U-shaped histograms at all sample sizes, so at the levels typically used in practice, the null hypothesis gets rejected more often than warranted.  

\begin{figure}[tb]  
\centering
\includegraphics[width = 0.95 \textwidth]{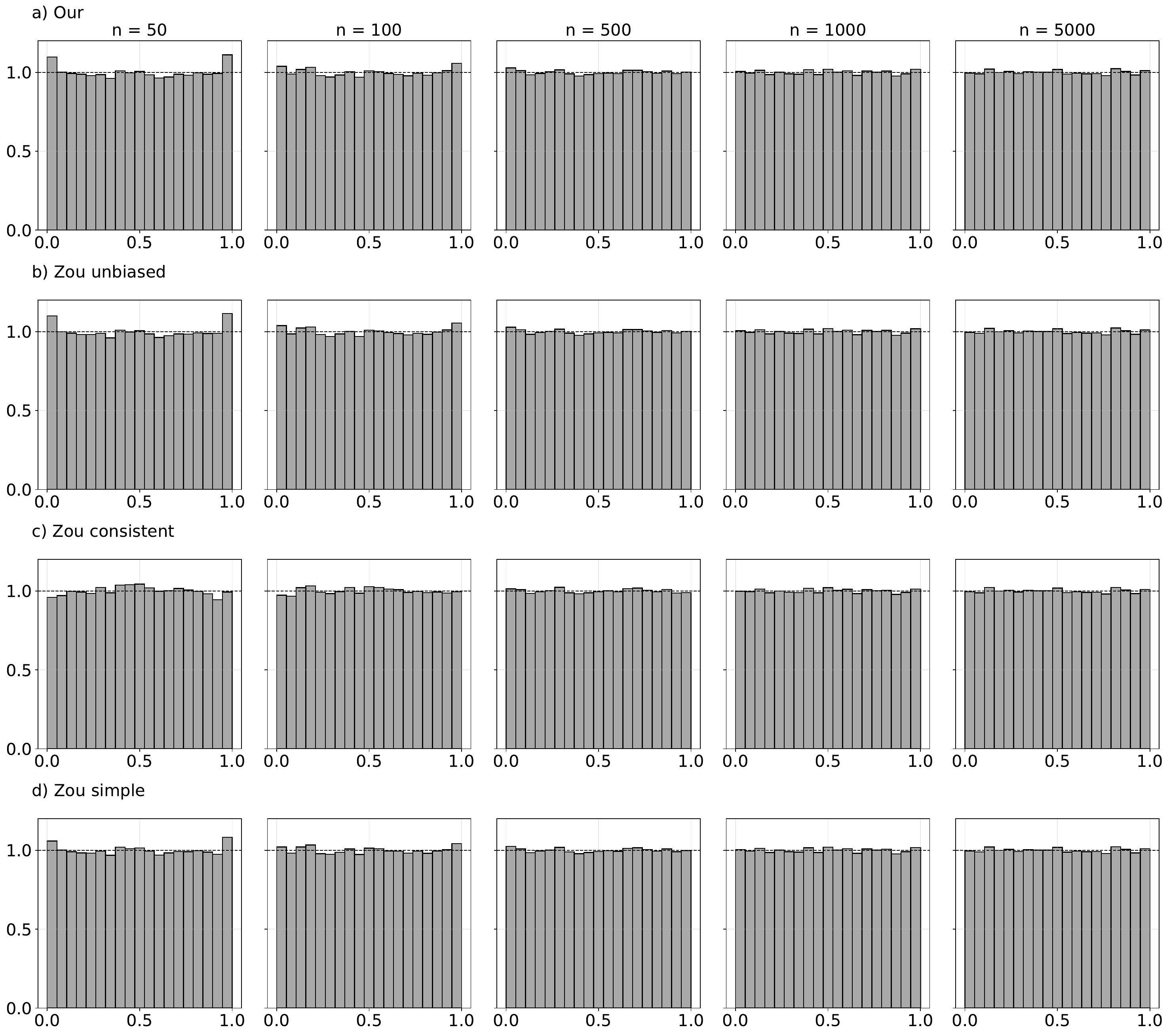}
\caption{Histograms of one-sided $p$-values in tests of equal $\AKC$ in the second, discrete simulation example, using the general test of DeLong type with test statistic at \eqref{eq:pairwise} and a) the plug-in variance estimator from our Section~\ref{sec:asymptotic}, and b) the unbiased, c) the consistent, and d) the simple variance estimator from \citet{Zou2022}.  The sample size $n$ equals $50$, $100$, $500$, $1000$, and $5000$ (left to right), and the number of Monte Carlo replicates is $100,000$.  \label{fig:sim_2b}}
\end{figure}

In the general, possibly discrete case the $\AGC$ and $\CMA$ measures have not been studied before, so we are unable to compare to extant methods.  Tending to our second, discrete simulation setting, Figure \ref{fig:sim_2a} demonstrates that the test for equal $\AGC$ with test statistic at \eqref{eq:pairwise} yields $p$-values that are essentially uniform, even at modest sample sizes.  For the $\AKC$ and $\CID$ measures, the literature provides a comprehensive set of variance estimators that apply in the general case.  In particular, \citet[Section~2.2]{Zou2022} review what they call the unbiased, the consistent, and the simple variance estimator, respectively, for a difference in these quantities.  Figure \ref{fig:sim_2b} shows histograms of one-sided $p$-values for the test of equal $\AKC$ with test statistics at \eqref{eq:pairwise} in the second simulation setting, where we compare the variant with our plug-in variance estimator (top row) to  variants with the estimators studied by \citet{Zou2022}.  While there are differences at the smallest sample size, where the consistent estimator stands out, all variants yield essentially uniform $p$-values when $n \geq 100$.

\section{Case studies}  \label{sec:case_studies}

We now showcase the use of the asymmetric measures in case studies on topical scientific issues.  First, we consider a recent study of \citet{Huang2024} and assess methods for uncertainty quantification in the output of large language models (LLMs).  Then we follow up on the scientific revolution in weather prediction \citep{Bi2023, Lam2023, BenBouallegue2024}, where data-driven methods are replacing previously used physics-based approaches.  In a study of the WeatherBench 2 benchmark dataset \citep{Rasp2024}, we shed new light on the predictive abilities of data-driven versus physics-based precipitation forecasts.

Both case studies concern the evaluation of predictive performance, where measures of monotone association with the crucial properties (e) and (f) from Theorems \ref{thm:AKC_properties} and \ref{thm:AGC_properties} are frequently preferred.  The perfect monotone predictor property (e) ensures that a potential predictor variable receives a perfect score if, and only if, the outcome is a nondecreasing, deterministic function of the predictor variable.  Invariance under strictly increasing transformations of the measures, as formalized by property (f), ensures that the measure quantifies discrimination ability, or potential predictive ability, as opposed to measures like squared error that target actual predictive ability, a task for which we refer to \citet{Gneiting2011}.  Evidently, both measures of potential predictive ability and actual predictive ability have important and ubiquitous usages.  Preferences for the former typically rely on the fact that discrimination ability can be improved by postprocessing, as aptly voiced by \cite[p.~2119]{Antolini2004}: 
\begin{quote}
\small
``[d]iscrimination is recognized as the first priority in judging prediction accuracy, since unlike the calibration, that can be improved through recalibration [\ldots] discrimination cannot be improved by any adjustment''.
\end{quote}
As reasoned in Section~\ref{sec:choice}, we prefer $\CMA$ and $\CID$ over $\AGC$ and $\AKC$ in the case studies, due to their immediate link to the $\AUC$ criterion that has been commonly used in the respective scientific communities.  Furthermore, we focus on $\CMA$, due to its lesser dependence on possibly artificial discretization levels in outcomes.

\subsection{Evaluating uncertainty metrics for large language models}  \label{sec:LLM}

Recent breakthroughs in large language models (LLMs) have prompted the development of methods for uncertainty quantification for their outputs, aiming to mitigate hallucinations and improve interpretability and usability \citep{Shorinwa2025}.  Uncertainty quantification for the output of LLMs bears challenges that hinder the direct application of evaluation methods developed for classification and regression tasks.  To address the specific requirements in this setting, \citet{Huang2024} introduce a rank calibration framework and propose a new measure, the rank calibration error ($\RCE$), to evaluate uncertainty or confidence measures for LLMs.  

In their study, \citet{Huang2024} consider five uncertainty metrics that map a query and the associated response generated by the LLM to a quantitative judgement of uncertainty or confidence --- namely, the $U_\textrm{Ecc}$, $U_\textrm{Deg}$, $U_\textrm{EigV}$, $U_\textrm{NLL}$, and $U_\textrm{SE}$ metrics, respectively --- and use $\RCE$ to quantify the alignment of the measure with one of four considered a posteriori correctness scores --- namely, BERT, METEOR, ROUGE-L, and ROUGE-1, respectively.\footnote{For detailed information on the uncertainty metrics and correctness scores, we refer to \citet{Huang2024} and references therein.}  We adopt the experimental setup and data of \citet{Huang2024} and evaluate the five measures relative to each of the four scores, based on responses generated by the Llama-2-7b and Llama-2-7b-chat \citep{Touvron2023} and GPT-3.5-turbo \citep{Ouyang2022} models on NQ-open \citep[][$n = 3,600$]{Lee2019}, SQuAD \citep[][$n = 10,570$]{Rajpurkar2016}, and TriviaQA \citep[][$n = 11,322$]{Joshi2017} queries, respectively.  The key tenet in \citet{Huang2024} is that higher confidence (respectively, lower uncertainty) ought to imply higher generation quality.  Given an uncertainty metric and a correctness score, the negatively oriented $\RCE$ measure (the lower, the better) quantifies deviations from this key assumption. 

\renewcommand{\arraystretch}{1.0}
\begin{table}[tp] 
\centering
\caption{Results in terms of $\CMA$ for experimental configurations in the format of Table 2 in \citet{Huang2024}, with temperature values set at .6 for Llama-2-7b and Llama-2-7b-chat and 1.0 for GPT-3.5, respectively.  For every value larger than .500 in the table, our one-sided test rejects the hypothesis of a $\CMA$ value less than or equal to $\frac{1}{2}$.  The highest $\CMA$ value in each row is shown in \textcolor{mycolor}{green} color.  The subscripts to these values indicate the result ($\CMA_+$: rejection; $\CMA_0$: no rejection) of a one-sided test of the hypothesis of equality with the second-highest value in the row.  All tests are at level .01.  \label{tab:LLM}}
\begin{tabular}{llllllll}
\toprule
Model & Queries & Correctness & $U_\textrm{Ecc}$ & $U_\textrm{Deg}$ & $U_\textrm{EigV}$ & $U_\textrm{NLL}$ & $U_\textrm{SE}$ \\
\midrule
\multirow{12}{*}{Llama-2-7b} & \multirow{4}{*}{NQ-open} & BERT & $.526$ & \color{mycolor}{$.653_+$} & $.650$ & $.545$ & $.565$ \\
& & METEOR & $.522$ & \color{mycolor}{$.633_+$} & $.626$ & $.521$ & $.565$ \\
& & ROUGE-L & $.530$ & \color{mycolor}{$.656_+$} & $.649$ & $.528$ & $.578$ \\
& & ROUGE-1 & $.530$ & \color{mycolor}{$.656_+$} & $.649$ & $.528$ & $.578$ \\
\cmidrule{3-8}
& \multirow{4}{*}{SQuAD} & BERT & $.510$ & $.624$ & \color{mycolor}{$.626_+$} & $.605$ & $.530$ \\
& & METEOR & \color{mycolor}{$.504_+$} & $.458$ & $.458$ & $.411$ & $.443$ \\
& & ROUGE-L & \color{mycolor}{$.497_+$} & $.453$ & $.452$ & $.426$ & $.455$ \\
& & ROUGE-1 & \color{mycolor}{$.497_+$} & $.452$ & $.452$ & $.426$ & $.455$ \\
\cmidrule{3-8}
& \multirow{4}{*}{TriviaQA} & BERT & $.531$ & \color{mycolor}{$.718_+$} & $.715$ & $.558$ & $.646$ \\
& & METEOR & $.515$ & $.638$ & $.631$ & $.529$ & \color{mycolor}{$.653_+$} \\
& & ROUGE-L & $.528$ & \color{mycolor}{$.700_+$} & $.694$ & $.547$ & $.672$ \\
& & ROUGE-1 & $.528$ & \color{mycolor}{$.700_+$} & $.694$ & $.546$ & $.672$ \\
\midrule
\multirow{12}{*}{Llama-2-7b-chat} & \multirow{4}{*}{NQ-open} & BERT & $.610$ & $.684$ & \color{mycolor}{$.686_+$} & $.576$ & $.673$ \\
& & METEOR & $.606$ & $.656$ & $.658$ & $.549$ & \color{mycolor}{$.672_+$} \\
& & ROUGE-L & $.620$ & $.677$ & $.680$ & $.554$ & \color{mycolor}{$.680_0$} \\
& & ROUGE-1 & $.620$ & $.677$ & $.680$ & $.554$ & \color{mycolor}{$.680_0$} \\
\cmidrule{3-8}
& \multirow{4}{*}{SQuAD} & BERT & $.547$ & $.578$ & $.584$ & \color{mycolor}{$.614_+$} & $.586$ \\
& & METEOR & $.493$ & $.433$ & $.437$ & \color{mycolor}{$.553_+$} & $.520$ \\
& & ROUGE-L & $.475$ & $.423$ & $.429$ & \color{mycolor}{$.559_+$} & $.513$ \\
& & ROUGE-1 & $.475$ & $.423$ & $.429$ & \color{mycolor}{$.558_+$} & $.513$ \\
\cmidrule{3-8}
& \multirow{4}{*}{TriviaQA} & BERT & $.719$ & $.769$ & $.767$ & $.767$ & \color{mycolor}{$.775_+$} \\
& & METEOR & $.670$ & $.706$ & $.704$ & $.671$ & \color{mycolor}{$.712_+$} \\
& & ROUGE-L & $.721$ & $.765$ & $.763$ & $.750$ & \color{mycolor}{$.777_+$} \\
& & ROUGE-1 & $.721$ & $.765$ & $.763$ & $.750$ & \color{mycolor}{$.777_+$} \\
\midrule
\multirow{12}{*}{GPT-3.5} & \multirow{4}{*}{NQ-open} & BERT & $.689$ & $.765$ & $.765$ & \color{mycolor}{$.775_+$} & $.761$ \\
& & METEOR & $.660$ & $.716$ & \color{mycolor}{$.717_0$} & $.683$ & $.707$ \\
& & ROUGE-L & $.689$ & \color{mycolor}{$.755_0$} & $.755$ & $.739$ & $.752$ \\
& & ROUGE-1 & $.689$ & \color{mycolor}{$.755_0$} & $.755$ & $.739$ & $.752$ \\
\cmidrule{3-8}
& \multirow{4}{*}{SQuAD} & BERT & $.521$ & $.531$ & $.533$ & \color{mycolor}{$.576_+$} & $.562$ \\
& & METEOR & $.544$ & $.567$ & $.567$ & \color{mycolor}{$.618_+$} & $.612$ \\
& & ROUGE-L & $.549$ & $.578$ & $.578$ & \color{mycolor}{$.650_+$} & $.640$ \\
& & ROUGE-1 & $.544$ & $.571$ & $.571$ & \color{mycolor}{$.640_+$} & $.629$ \\
\cmidrule{3-8}
& \multirow{4}{*}{TriviaQA} & BERT & $.682$ & $.705$ & $.704$ & \color{mycolor}{$.715_0$} & $.713$ \\
& & METEOR & $.662$ & \color{mycolor}{$.681_+$} & $.681$ & $.621$ & $.645$ \\
& & ROUGE-L & $.749$ & $.780$ & $.778$ & $.783$ & \color{mycolor}{$.791_+$} \\
& & ROUGE-1 & $.750$ & $.781$ & $.779$ & $.782$ & \color{mycolor}{$.790_+$} \\
\bottomrule
\end{tabular}
\end{table}

\begin{figure}[t]
\centering
\includegraphics[width = 0.800 \textwidth]{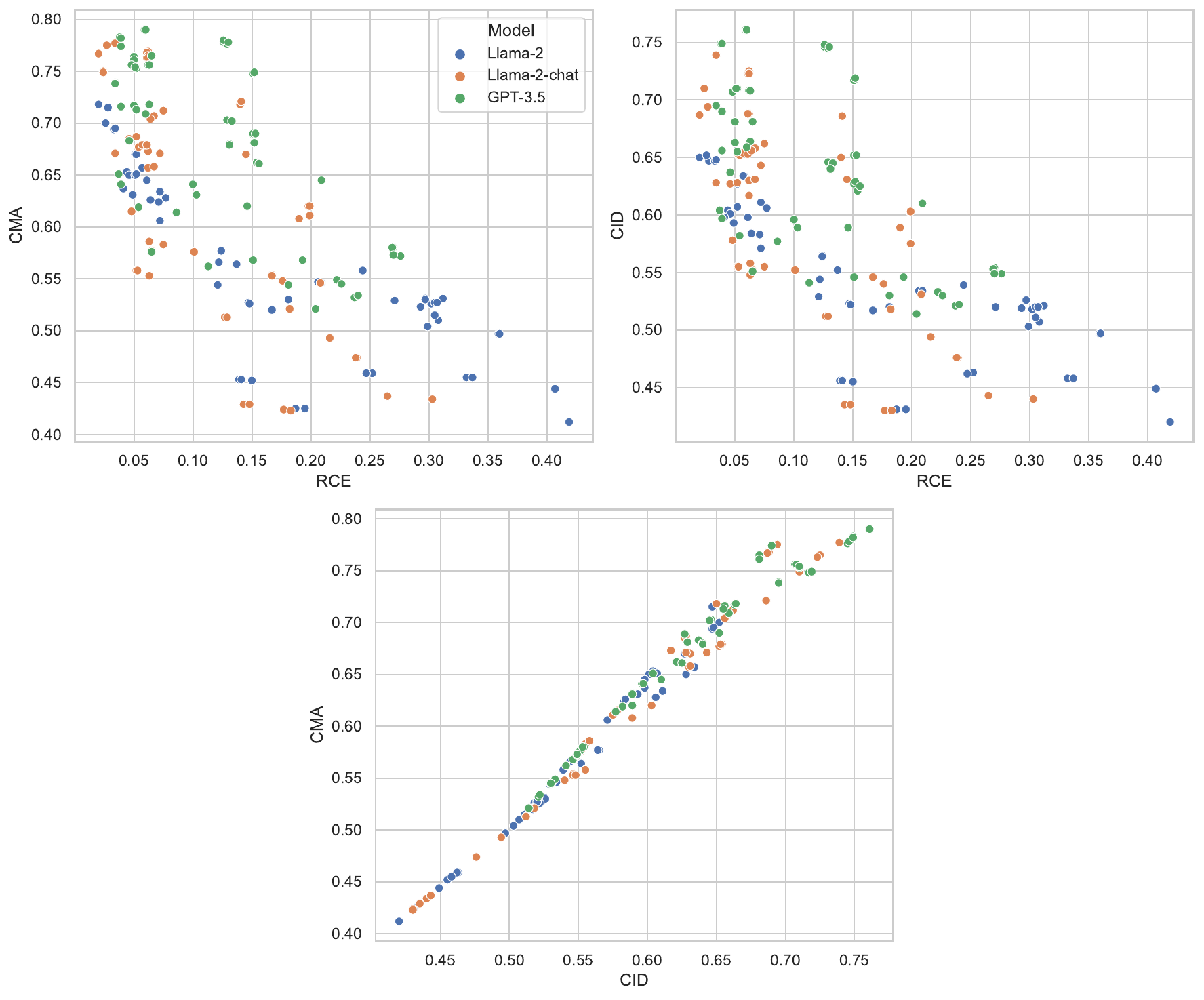}
\caption{Pairwise scatterplots of the positively oriented $\CMA$ and $\CID$ values, and the negatively oriented $\RCE$ measure, in the 180 experimental configurations from Table~\ref{tab:LLM}.  Following \citet{Huang2024} we display values averaged over 20 bootstrap samples.  \label{fig:LLM}}
\end{figure}

\citet{Huang2024} argue that $\RCE$ avoids thresholding of the correctness score and applies to any type of output range for the uncertainty metric.  The positively oriented $\CMA$ and $\CID$ measures share these desirable properties.  Furthermore, if the correctness values are dichotomous at the outset, $\CMA$ and $\CID$ reduce to the $\AUC$ measure at \eqref{eq:AUC}, which has been widely used in the LLM research community \citep{Kuhn2023, Lin2024}.

While $\RCE$ was developed specifically for the task at hand, we now compare to the general purpose measures $\CMA$ and $\CID$ that quantify monotone association, with a null value of $\frac{1}{2}$ under random assignments.  Specifically, Table \ref{tab:LLM} shows $\CMA(-X,Y)$ in 180 experimental configurations, where $X$ is one of the five uncertainty metrics and $Y$ is one of the four correctness scores, for each of the three language models and three query sources.  We note a certain robustness to the choice of the correctness score; in particular, the results under the ROUGE-L and the ROUGE-1 scores are very similar.  For the Llama-2-7b and Llama-2-7b-chat models and responses to SQuAD queries various $\CMA$ values are below $\frac{1}{2}$, whereas for NQ-open and TriviaQA queries the $\CMA$ value is above $\frac{1}{2}$ for all $3 \times 2 \times 5 \times 4 = 120$ configurations of LLM, query source, uncertainty metric, and correctness score, respectively.  The markedly different pattern observed for SQuAD with many $\CMA$ values near or even below $\frac{1}{2}$ may stem from fundamental differences in task structure compared to the open-domain question answering format of NQ-open and TriviaQA.  The reading comprehension format of SQuAD creates uncertainty not only about factual content, but also related to response exactitude.  We encourage the future exploration of task-dependent performance patterns to inform evaluation practices in the LLM research community.

Based on the tests developed in Section~\ref{sec:sample} statistical significance can be assessed.  Specifically, we test the hypothesis that $\CMA$ equals the null value of $\frac{1}{2}$ based on the test statistic at \eqref{eq:simple}, and we consider pairwise comparisons based on the test statistic at \eqref{eq:pairwise}.  In tests of $H_0 : \CMA \leq \frac{1}{2}$ at level $.01$ every single $\CMA$ value larger than .500 in Table \ref{tab:LLM} entails rejection.  The results of one-sided tests of equal $\CMA$ between the uncertainty metrics with the highest and the second-highest $\CMA$ value are documented in the table.  The respective table of $\CID$ values (not shown) features essentially the same results in terms of these tests.

The scatterplots in Figure \ref{fig:LLM} cover $\CMA$ and $\CID$ along with $\RCE$, where we follow the computation of the $\RCE$ values in \citet{Huang2024} and average over 20 bootstrap samples.  The all-purpose measures $\CMA$ and $\CID$ quantify the correctness values' ability to discriminate between queries of varying difficulty, which relates loosely only to the facets of predictive performance addressed by $\RCE$.  While $\CMA$ and $\CID$ share desirable properties of $\RCE$, they seem easier to use and interpret, due to the natural threshold of $\frac{1}{2}$ that corresponds to randomly assigned judgments of uncertainty, and the reduction to $\AUC$ when the correctness values are dichotomous in the first place.  Generally, $\CMA$ attains values higher than $\CID$, well in line with the properties of Spearman's Rho and Kendall's Tau in continuous settings \citep{Fredricks2007}.

\subsection{WeatherBench 2: Data-driven vs physics-based weather prediction}  \label{sec:WeatherBench}

The past few years have witnessed dramatic change in the practice of weather prediction.  While until a few years ago, weather forecasts had relied on physics-based numerical weather prediction (NWP) models, recent advances in neural network methodologies have enabled the rise of purely data-driven, artificial intelligence (AI) based weather prediction \citep[AIWP;][]{Bi2023, Lam2023} models.  Arguably, by now AIWP models are outperforming NWP models \citep{BenBouallegue2024, Gneiting2025}.  However, in contrast to temperature, pressure, and wind speed, forecasts of precipitation accumulations have received scant attention only in the comparison of physics-based and data-driven weather forecasts, with the notable exception of \citet{Radford2025}.  The reasons for the neglect of precipitation forecasts are two-fold.  On the one hand, researchers typically rely on the ERA5 product \citep{Hersbach2020} for ground truth data, and the quality thereof can be poor for precipitation, particularly towards polar regions, where precipitation subsumes rainfall, hail, and snow \citep{Lavers2022}.  Furthermore, precipitation is a mixed discrete-continuous variable with a point mass at zero (when there is no precipitation) and a skewed and heavy right tail, which challenges the application of standard evaluation measures.  However, the $\CMA$ and $\CID$ measures adapt naturally to the mixed discrete-continuous character of precipitation accumulations.  

In this light, we now consider the WeatherBench 2 benchmark dataset \citep{Rasp2024} and use the $\CMA$ and $\CID$ measures to compare NWP and AIWP forecasts of 24-hour precipitation accumulation at a prediction horizon of a day ahead.  We restrict the analysis to the premier NWP model, namely, the European Centre for Medium Range Weather Forecasts (ECMWF) high-resolution (HRES) model, and a leading AIWP model, namely the GraphCast model \citep{Lam2023} in its operational version,\footnote{These models are denoted IFS HRES and GraphCast (oper.)~in the WeatherBench 2 dataset \citep{Rasp2024}, which is available at \url{https://sites.research.google/gr/weatherbench/}.} which are run with the same initial conditions, thus allowing for a fair comparison.  The spatial resolution is 1.5 degrees, so each latitude band comprises 240 grid cells, for which forecasts were generated twice daily (in retrospect, for 2020) with initialization times 12 hours apart.  For reference, we compare the HRES and GraphCast models to the static WeatherBench 2 climatology forecast, which does not invoke numerical, statistical, or machine learning methods, and to the simplistic persistence forecast, which projects the most recent observed precipitation forward.  The aforementioned ERA5 product \citep{Hersbach2020} serves to provide ground truth.

We follow up on the analysis in Section~5.3 of \citet{Gneiting2022b} for the earlier WeatherBench 1 dataset \citep{Rasp2020}, where NWP forecasts outperformed AIWP forecasts by huge margins.  As in the earlier study, we use a range of performance metrics, which includes $\CMA$ skill and $\CID$ skill in the form of the percentage improvement for the forecast at hand over the climatology forecast.  Furthermore, we consider root mean squared error (RMSE) skill relative to the climatology forecast, stable equitable error in probability space \citep[SEEPS;][]{Rodwell2010} skill relative to the climatology forecast, and the anomaly correlation coefficient (ACC) in the form described by \citet[][Section~4.3.2]{Rasp2024}.  All five measures are positively oriented --- the larger, the better --- and by definition their value equals zero for the climatology forecast.\footnote{We average measures over latitude bands and then compute skill.  In passing, we note that the ideas behind the representation of $\AGC$ and $\CMA$ in \eqref{eq:AGC}, \eqref{eq:AGC_granularity}, and \eqref{eq:CMA2s}, respectively, resemble those put forth by meteorologists in the context of forecast evaluation for precipitation forecasts, where the linear error in probability space (LEPS) and SEEPS measures have been proposed \citep{Rodwell2010}.  However, unlike SEEPS, $\AGC$ and $\CMA$ do not require researchers to artificially discretize their data.  We also recall that RMSE is a measure of actual predictive performance, whereas ACC, $\AGC$, and $\CMA$ can be interpreted as measures of discrimination ability and, therefore, potential predictive ability.}

\begin{figure}[t]
\centering
\includegraphics[width = 0.95 \textwidth]{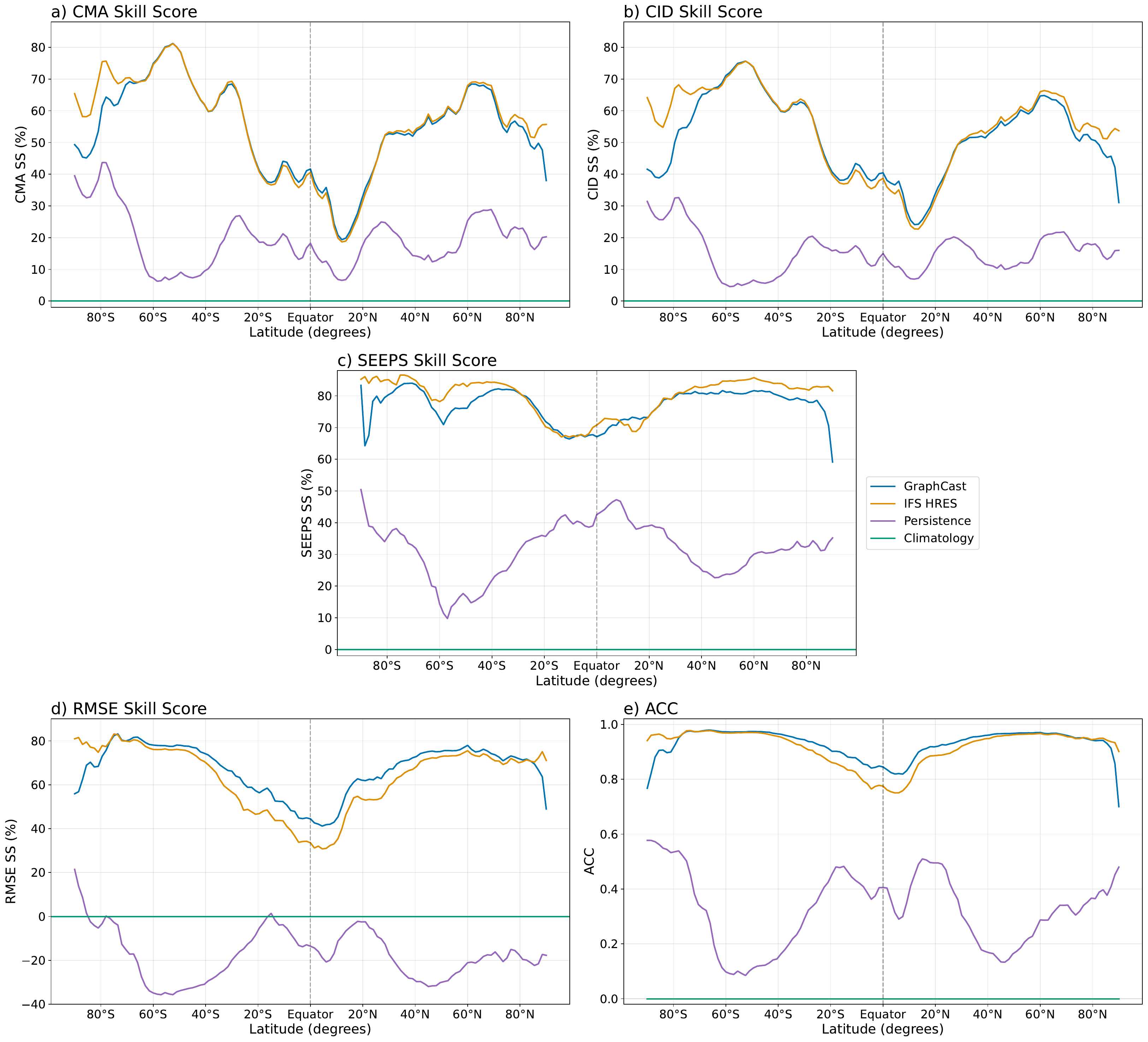}
\caption{Predictive ability of WeatherBench 2 HRES, GraphCast, and persistence forecasts of 24-hour precipitation accumulation relative to the climatology forecast in terms of a) $\CMA$ skill, b) $\CID$ skill, c) RMSE skill, d) ACC, and e) SEEPS skill.  The performance measures are averaged over the 240 grid cells in each latitude band prior to computing skill.  \label{fig:WeatherBench2}}
\end{figure}

Figure \ref{fig:WeatherBench2} shows the performance measures in their dependence on latitude bands at a resolution of 1.5 degrees.  Each latitude band has 240 grid cells with forecasts twice daily in the 2020 evaluation period, for a total of 732 forecasts per grid cell.  We compute the performance measures grid cell by grid cell and then average across the 240 grid cells in the latitude band.  The climatology forecast is unable to capture day-to-day variability but tends to have smaller errors than the persistence forecast, which extrapolates the precipitation patterns at the grid cell at hand, but magnifies magnitude errors when weather systems move or intensify swiftly.  Therefore, climatology outperforms the persistence forecast in terms of RMSE, and vice versa in terms of ACC, SEEPS, $\CMA$, and $\CID$.  The GraphCast and HRES forecasts outperform the reference forecasts by huge margins.  Interestingly, the data-driven GraphCast forecast has a competitive edge over the physics-based HRES forecast in the tropics near the equator, and vice versa towards the poles.  However, details depend on the performance measure used, and the differences between the GraphCast and HRES forecasts pale relative to the differences to the reference forecasts.

\begin{figure}[t]
\centering
\includegraphics[width = 1.000 \textwidth]{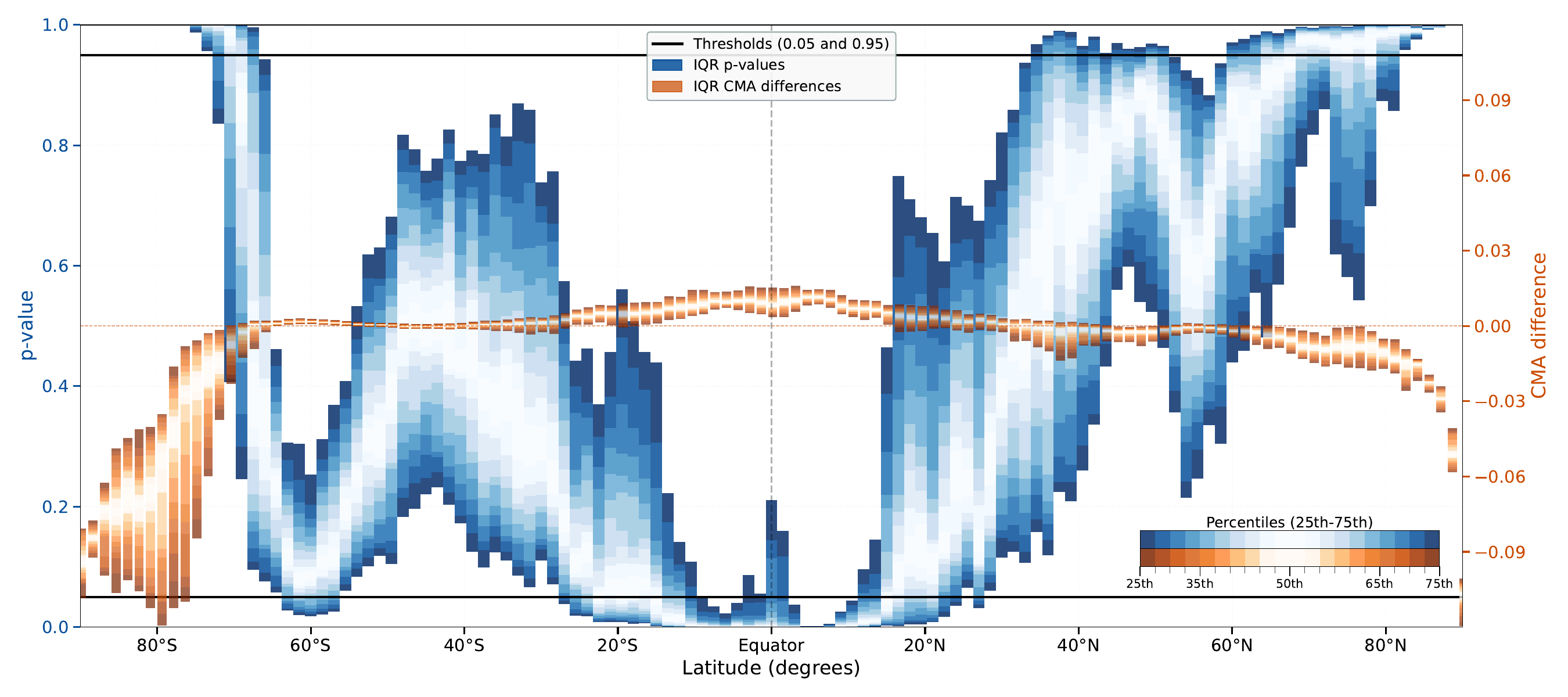}
\caption{Comparison of WeatherBench 2 GraphCast and HRES forecasts of 24-hour precipitation accumulation by latitude bands.  The boxplots show the interquartile range (IQR) of the difference in $\CMA$ in shades of red, and the associated one-sided $p$-value for the test statistic at \eqref{eq:pairwise} in shades of blue, considering the 240 grid cells in each latitude band, respectively.  Positive $\CMA$ differences and small $p$-values indicate superiority of the GraphCast model over the HRES model, and vice versa.  \label{fig:WeatherBench_pvals}}
\end{figure}

In this light, Figure \ref{fig:WeatherBench_pvals} provides a more nuanced comparison of the predictive performance of the WeatherBench 2 GraphCast and HRES forecasts in terms of $\CMA$.  As noted, each latitude band comprises 240 grid cells, and for each latitude band we show a boxplot of the respective differences in $\CMA$, with positive values indicating superiority of the GraphCast model.  WeatherBench 2 comprises two daily runs with initialization times 12 hours apart, so successive forecasts and outcomes for 24-hour precipitation accumulation must be assumed to be dependent.  Therefore, we use HAC estimates in the denominator of the test statistic at \eqref{eq:pairwise}, with implementation details as proposed by \citet[Section~5.3]{Pohle2025a}.  The boxplots for the respective one-sided $p$-values show that around the equator GraphCast has a competitive edge over HRES, whereas HRES outperforms GraphCast towards the poles.  The respective boxplots for $p$-values based on the $\CID$ measure look essentially identical and we omit them.  

These findings are striking in view of the emerging consensus that the latest data-driven weather models outperform physics-based models.  However, they are broadly compatible with a note in \citet[p.~869]{BenBouallegue2024}, who report that for pressure the otherwise competitive, data-driven Pangu-Weather model of \citet{Bi2023} performs much worse over the central Arctic than the HRES model.  One possible explanation is the use of cosine-latitude weighted scores when training the data-driven models \citep{BenBouallegue2024}.  A complementary or alternative explanation lies in a possible neglect of the tropics in the decade-long development of physics-based NWP models, with almost exclusive involvement of researchers residing at mid-latitude locations.  We find it tempting to speculate that, if further development of NWP models were to target precipitation in the tropics, while the training of AIWP models was to give more emphasis to subpolar and polar regions, then the predictive ability of AIWP models and NWP models would be roughly on par globally, and rather close to predictability limits, given current observational assets.

\section{Discussion}  \label{sec:discussion} 

Across scientific fields, researchers have sought measures of association that are invariant under strictly increasing transformations of the margins.  In this paper, we started from the observation that the literature on such measures has been splintered.  For continuous variables, symmetric rank correlation coefficients, such as Spearman's Rho and Kendall's Tau, have been studied in detail in the statistical literature.  For dichotomous outcomes, the asymmetric area under the curve ($\AUC$) measure has been used to assess monotone dependence.  Our population level discussion illuminates the distinction between symmetric and asymmetric measures.  Symmetric measures are subject to the attainability problem, and as one leaves the realm of continuous variables, they fail to satisfy the perfect monotone predictor property.  As scientific communities move to ever more automated analyses, it seems useful to work with measures that adapt naturally to the discrete, continuous, or mixed discrete-continuous character of the variables involved.  In this light, we join \citet{Chatterjee2021}, \citet{Azadkia2025}, and \citet{Roudaki2026} in arguing that, in much of statistical practice and data analysis, a reorientation from symmetric to asymmetric measures of dependence might be warranted.  

As the attainability problem is linked to discrete components, an appealing approach is to consider asymmetric measures of monotone association that cover all types of variables, but specialize to Spearman's Rho or Kendall's Tau in the continuous case.  Furthermore, the asymmetric measure ought to reduce to the ubiquitously used $\AUC$ measure if the outcome is dichotomous.  We have studied asymmetric measures of monotone association that satisfy these requirements and accomplish unified ways of quantifying monotone dependence for dichotomous, ranked categorical, mixed discrete-continuous, and continuous types of outcomes. 

The asymmetric Kendall correlation ($\AKC$) or Somers' $D$ at \eqref{eq:AKC} and the associated C index or concordance index ($\CID$) at \eqref{eq:CID} and \eqref{eq:CID_AKC} bridge Kendall's Tau and $\AUC$.  Following the landmark paper by \citet{Harrell1996}, the $\CID$ measure has been vastly popular in survival analysis, where censored outcomes are ubiquitous.  We emphasize that recent critiques of the $\CID$ measure \citep{Blanche2019, Hartman2023, Lillelund2026} apply to its use specifically in the setting of survival analysis, where researchers assess predictive distributions relative to censored outcomes, but not to the usages discussed here, where one assesses covariate, feature, or predictor variables relative to fully observed outcomes.  In fact, the emphasis on, and the intricacies of, censored outcomes may have hindered appreciation and dissemination of the $\AKC$ and $\CID$ measures in other fields, where censored outcomes may not be an issue at all.

The key original contribution of this paper lies in the introduction and analysis of asymmetric grade correlation ($\AGC$) at \eqref{eq:AGC} and the coefficient of monotone association ($\CMA$) measure at \eqref{eq:CMA}.  These measures bridge Spearman's Rho and $\AUC$ in the very same way that $\AKC$ and $\CID$ bridge Kendall's Tau and $\AUC$.  Faced with a choice between $\AGC$ and $\CMA$, and $\AKC$ and $\CID$, we might voice a slight preference for $\AGC$ and $\CMA$, due to the lesser dependence on the discretization level of the outcomes.

The associated large sample theory applies to all these measures and allows for the development of standard tools for asymptotic inference.  In this way, our results unify and complete thus far disconnected strands of literature, by establishing common population level theory, common estimators, and common tests that bridge continuous and dichotomous settings and apply to all linearly ordered outcomes.  In particular, our generalization of the \citet{DeLong1988} test for pairwise comparisons of $\AUC$ values for a dichotomous outcome accomplishes pairwise comparisons under all types of linearly ordered outcomes.  Software for our estimators and tests is available in Python \citep{Python}, and we refer the reader to \url{https://github.com/evwalz/paper_monotone_dependence} for code and reproduction materials.

Our plea for the use of asymmetric measures of monotone association seems complete with population level theory, large sample theory, tools for inference, case studies, and code --- with the development being original in the case of $\AGC$ and $\CMA$, and expository in the case of $\AKC$ and $\CID$.  Nonetheless, our paper leaves considerable scope for follow-up work.  On the infrastructural side, we intend to develop our code into software packages for both Python and \textsf{R} \citep{R}.  Tending to theoretical and methodological tasks, measures of monotone association are invariant under strictly increasing transformations of the margins, but their values may improve under non-strictly increasing mappings of the covariate or predictor.  For dichotomous outcomes, it is well-known that the pool-adjacent-violators (PAV) algorithm yields non-strictly increasing transformations that optimize $\AUC$ \citep{Fawcett2007}, and we are keen to learn about similar results for our measures and more general types of targets.  The limit distributions in the central limit theorems for the empirical versions of the measures are nondegenerate in typical practice.  However, degeneracies may occur, as analyzed in the case of Kendall's Tau by \citet[Section~2.4]{Dengler2010} and in the case of $\AUC$ by \citet{Demler2012} and \citet{Demler2017}, who study the effects of nested models and adjustments for estimated parameters.  We also encourage the development of large sample theory and associated tools for statistical inference in the general setting at \eqref{eq:data_MV_generalized}, where issues of degeneracy amplify.  Finally, a formidable challenge lies in the development of measures of monotone association and associated methodology for inference when the predictor variables and/or outcomes are multivariate, perhaps along the lines of recent work by \citet{Hallin2021}.

\section*{Acknowledgements}

We have used Cursor to aid code development.  We thank Marc-Oliver Pohle for numerous discussions and generously sharing a preliminary version of \citet{Pohle2025a}, which supplied the technical tools in the proofs of the asymptotic results in our Section~\ref{sec:sample}.  We have furthermore benefited from discussions with Olga Demler, Marc Genton, Alexander Jordan, Florian Kalinke, Peter Knippertz, Kristof Kraus, Sebastian Lerch, Marc Strickert, Jan St{\"u}hmer, and Lu Yang, and from comments by anonymous reviewers of an earlier version of the manuscript.  Eva-Maria Walz and Tilmann Gneiting are grateful for support by the Klaus Tschira Foundation.

\appendix 

\renewcommand\theequation{\thesection.\arabic{equation}}
\renewcommand\thefigure{\thesection.\arabic{figure}}

\section{Proofs for Section~\ref{sec:population}}  \label{app:proofs_population}

\setcounter{equation}{0} 
\setcounter{figure}{0} 

\begin{proof}[Proof of Theorem \ref{thm:AKC_properties}]
The claims in parts (a), (b), (c), (d), and (f) are immediate from the relations at~\eqref{eq:tau}, \eqref{eq:AKC}, and \eqref{eq:AKC_granularity}.  We proceed to prove part (e).  First, suppose that $Y = m(X)$ almost surely, where $m$ is nondecreasing.  Then $Y' < Y''$ implies $X' < X''$ almost surely, and we conclude that 
\begin{align*}
\AKC(X,Y) = 2 \CID(X,Y) - 1 = 2 \, \myE(s(X',X'') \mid Y' < Y'') - 1 = 1.
\end{align*}
Conversely, suppose that $\AKC(X,Y) = \CID(X,Y) = 1$.  From \eqref{eq:CID}, $Y' < Y''$ implies $X' < X''$ almost surely, whereas $Y' > Y''$ implies $X' > X''$ almost surely.  Therefore, $X' = X''$ implies $Y' = Y''$ almost surely, so $Y = m(X)$ almost surely for some function $m$.  As $X' < X''$ together with $Y' > Y''$ occurs on a null set only, $m$ can be chosen to be nondecreasing.
\end{proof}

\begin{proof}[Proof of Theorem \ref{thm:CMA2s}.]
The definition of the $\AGC$ measure at \eqref{eq:AGC} implies that
\begin{align}  \label{eq:proof1}
\CMA(X,Y) 
= \frac{1}{2} \left( \frac{\cov(\Fbar(X),\Gbar(Y))}{\cov(\Gbar(Y),\Gbar(Y))} + 1 \right)
= \frac{\myE \! \left( \Fbar(X) \Gbar(Y) \right) + \myE \! \left( (\Gbar(Y))^2 \right) - \frac{1}{2}}{2 \hsp \myE \! \left( (\Gbar(Y))^2 \right) - \frac{1}{2}}.
\end{align}
To establish the equivalence of the representations in \eqref{eq:CMA2s} and \eqref{eq:proof1}, respectively, we consider enumerator and denominator separately.  For the enumerator, 
\begin{align*}
\myE \! & \left( (\Gbar(Y'') - \Gbar(Y')) \hsp \one \{Y' < Y''\} \hsp s(X',X'') \right) \\
& = \myE \! \left( (\Gbar(Y'') - \Gbar(Y')) \hsp s(Y',Y'') \hsp s(X',X'') \right) \\
& = \myE \! \left( \Gbar(Y') \hsp s(Y'',Y') \hsp s(X'',X') \right) - \myE \! \left( \Gbar(Y') \hsp s(Y',Y'') \hsp s(X',X'') \right) \\
& = \myE \! \left( \Fbar(X) \hsp \Gbar(Y) \right) + \myE \! \left( \Gbar(Y)^2 \right) - \textstyle \frac{1}{2},
\end{align*}
where the final equality follows upon conditioning on $(X',Y')$ and simplifying.  For the denominator,
\begin{align*}
\myE \! \left( (\Gbar(Y'') - \Gbar(Y')) \hsp \one\{Y' < Y''\} \right) 
= \myE \! \left( \Gbar(Y'') \hsp G(Y''-) \right) - \myE \! \left( \Gbar(Y') \hsp (1-G(Y')) \right) = 2 \, \myE \! \left( \Gbar(Y))^2 \right) - \textstyle \frac{1}{2},
\end{align*}
where $G(y-) = \lim_{z \uparrow y} G(z)$, with the first equality following upon conditioning on $Y''$ and $Y'$, respectively, and the second equality using the fact that $G(y-) + G(y) = 2 \, \Gbar(y)$ for $y \in \real$.
\end{proof}

\begin{proof}[Proof of Theorem \ref{thm:CMA2AUC}.]
Since $G$ is continuous, $\Gbar = G$, and we conclude from \eqref{eq:CMA2s} and \eqref{eq:AUC_alpha} in concert with Fubini's theorem that 
\begin{align*}
\CMA(X,Y) 
& = \frac{\myE \left[ (G(Y'') - G(Y')) \hsp \one\{Y' < Y''\} \hsp s(X',X'') \right]}{\myE \left[ (G(Y'') - G(Y')) \hsp \one\{Y' < Y''\} \right]} \\
& = \frac{\int \myE \left( s(X',X'') \hsp \one\{Y' < Y''\} \hsp \one\{Y' < y \leq Y''\} \right) \dint G(y)}{\int  \myE \left( \one\{Y' < Y''\} \hsp \one\{Y' < y \leq Y''\} \right) \dint G(y)} \\
& = \frac{\int \myE \left( s(X',X'') \hsp \one\{Y' < y\} \hsp \one\{Y'' \geq y\} \right) \dint G(y)}{\int \myE \left( \one\{Y' < y\} \hsp \one\{Y'' \geq y\} \right) \dint G(y)} \\
& = \frac{\int_0^1 \myE \left( s(X', X'') \hsp \one\{Y' < G^{-1}(\alpha)\} \hsp \one\{Y'' \geq G^{-1}(\alpha)\} \right) \dint \alpha}{\int_0^1 \myE \left( \one\{Y' < G^{-1}(\alpha)\} \hsp \one\{Y'' \geq G^{-1}(\alpha)\} \right) \dint \alpha} \\
& = \frac{\int_0^1 \myE \left( \one\{Y' < G^{-1}(\alpha), Y'' \geq G^{-1}(\alpha)\} \right) \myE \left( s(X', X'') \mid Y' < G^{-1}(\alpha), Y'' \geq G^{-1}(\alpha) \right) \dint \alpha}{\int_0^1 \alpha \hsp (1-\alpha) \dint \alpha} \\
& = 6 \int_0^1 \alpha \hsp (1-\alpha) \AUC^{\hsp (\alpha)}(X,Y) \dint \alpha,
\end{align*}
where the fourth equality follows upon substituting $\alpha = G(y)$.
\end{proof}

\begin{proof}[Proof of Corollary \ref{cor:CMA2AUC}.]
Immediate from Theorem \ref{thm:CMA2AUC} and \eqref{eq:CMA_continuous}.
\end{proof}

\begin{proof}[Details for Example \ref{ex:CMA2AUC}]
The joint distribution $\myP$ of $(X,Y)$ is uniform on the triangle $\textrm{T}$, as illustrated in Figure \ref{fig:T}a).  As $\myP$ and the marginal distributions $F = {\cal L}(X)$ and $G = {\cal L}(Y)$ are continuous, we may write
\begin{align*}
\AUC^{\hsp (\alpha)}(X,Y) = \myP(X_{0,\alpha} < X_{1,\alpha}),
\end{align*}
where $X_{0,\alpha}$ and $X_{1,\alpha}$ are independent with distribution equal to ${\cal L}(X \mid Y < G^{-1}(\alpha))$ and ${\cal L}(X \mid Y \geq G^{-1}(\alpha))$, respectively.

\begin{figure}[t]
\centering
\includegraphics[width = 1.000 \textwidth]{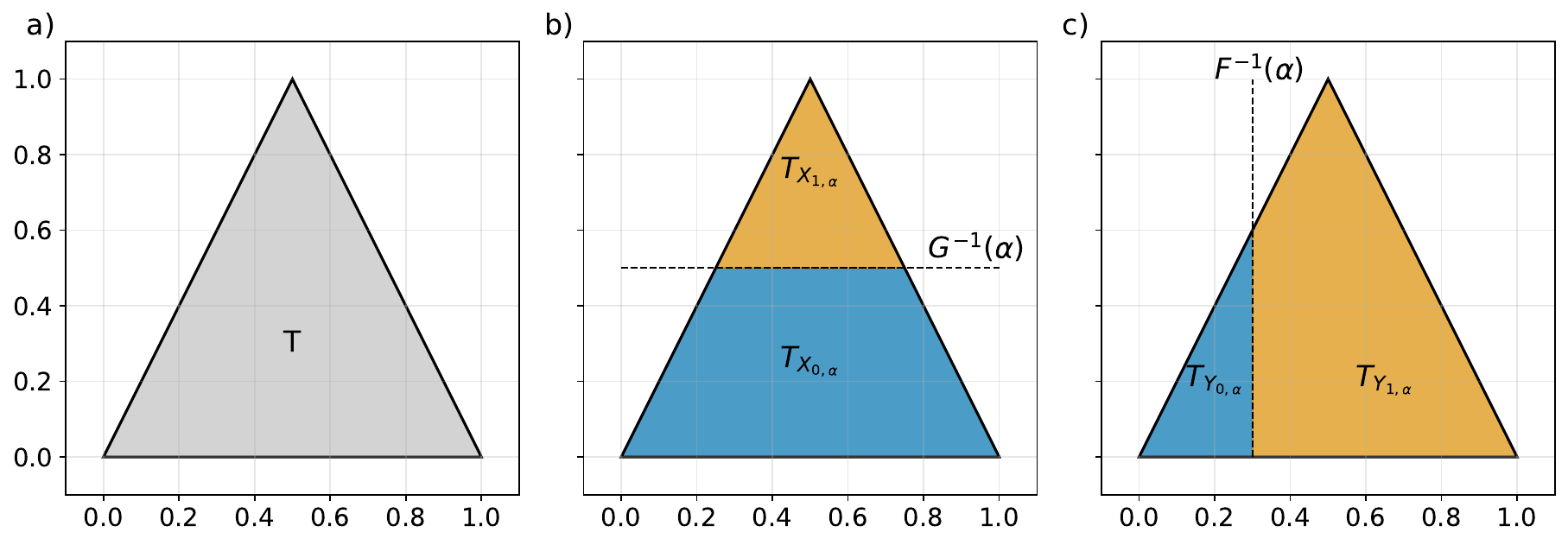}
\caption{The random vector $(X,Y)$ in Example \ref{ex:CMA2AUC} is uniformly distributed on a) the triangle $\textrm{T}$.  The b) horizontal and c) vertical divisions support the associated computations in this appendix.  \label{fig:T}}
\end{figure}

The quantity $\AUC^{\hsp (\alpha)}(X,Y)$ is straightforward to find for reasons of symmetry.  From Figure \ref{fig:T}b), the distributions of $X_{0,\alpha}$ and $X_{1,\alpha}$ are  symmetric about $\frac{1}{2}$, so
\begin{equation*}
\myP(X_{0,\alpha} < X_{1,\alpha}) = \myP(1 - X_{0,\alpha} < 1 - X_{1,\alpha}) = \myP(X_{0,\alpha} > X_{1,\alpha}).
\end{equation*}
Therefore, $\AUC^{\hsp (\alpha)}(X,Y) = \myP(X_{0,\alpha} < X_{1,\alpha}) = \frac{1}{2}$, as claimed.

Similarly, we may write
\begin{align*}
\AUC^{\hsp (\alpha)}(Y,X) = \myP(Y_{0,\alpha} < Y_{1,\alpha}),
\end{align*}
where $Y_{0,\alpha}$ and $Y_{1,\alpha}$ are independent with distribution equal to ${\cal L}(Y \mid X < F^{-1}(\alpha))$ and ${\cal L}(Y \mid X \geq F^{-1}(\alpha))$, respectively (Figure \ref{fig:T}c).  It suffices to consider $\alpha \in (0, \frac{1}{2})$, since
\begin{align}  \label{eq:AUC_Y_X_1}
\AUC^{\hsp (1 - \alpha)}(Y,X) = \myP(Y_{0,1-\alpha} < Y_{1,1-\alpha}) = 1 - \myP(Y_{0,\alpha} < Y_{1,\alpha}) = 1 - \AUC^{\hsp (\alpha)}(Y,X)
\end{align}
for reasons of symmetry.  If we let $x_\alpha = F^{-1}(\alpha)$, the densities of $Y_{0,\alpha}$ and $Y_{1,\alpha}$ are given by
\begin{align*}
g_{0,\alpha}(y) = \frac{1}{x_\alpha^2} \left( x_\alpha - \frac{y}{2} \right) \one \{0 \leq y \leq 2x_\alpha \}
\end{align*}
and 
\begin{align*}
g_{1,\alpha}(y) = \frac{1}{1 - 2x_\alpha^2} \left[ \left( 1 - x_\alpha - \frac{y}{2} \right) \one \{ 0 \leq y \leq 2x_\alpha \} + \left( 1 - y \right) \one \{ 2x_\alpha \leq y \leq 1 \} \right], 
\end{align*}
respectively.  Therefore,
\begin{align*}
\AUC^{\hsp (\alpha)}(Y, X) 
& = \myP(Y_{0,\alpha} < Y_{1,\alpha}) \\
& = \int_0^1 \int_0^1 \one \{ y_0 < y_1 \} \, g_{0,\alpha}(y_0) \, g_{1,\alpha}(y_1) \dint y_0 \dint y_1 \\
& = \int_0^{2x_\alpha} \int_0^{y_1} \frac{1}{x_\alpha^2} \left( x_\alpha - \frac{y_0}{2} \right) \dint y_0 \, \frac{2}{1 - 2x_\alpha^2} \left( 1 - x_\alpha - \frac{y_1}{2} \right) \dint y_1 + \int_{2x_\alpha}^1 \frac{2}{1 - 2x_\alpha^2} \left( 1 - y_1 \right) \dint y_1 \\
& = \frac{2}{1 - 2x_\alpha^2} \left[ \frac{1}{x_\alpha^2} \int_0^{2x_\alpha} \left( x_\alpha \hsp y_1 - \frac{y_1^2}{4} \right) \left( 1 - x_\alpha - \frac{y_1}{2} \right) \dint y_1 + \int_{2x_\alpha}^1 \left( 1 - y_1 \right) \dint y_1 \right] \\
& = \frac{2}{1 - 2x_\alpha^2} \left[ \frac{1}{x_\alpha^2} \int_0^{2x_\alpha} \left( \left( x_\alpha - x_\alpha^2 \right) y_1 - \frac{1 + x_\alpha}{4} y_1^2 + \frac{1}{8} y_1^3 \right) \dint y_1 + (1 - 2x_\alpha) - \frac{1 - 4x_\alpha^2}{2} \right] \\
& = 1 - \frac{\frac{4}{3} x_\alpha + \frac{1}{3} x_\alpha^2}{1 - 2x_\alpha^2}.
\end{align*}
To find $x_\alpha = F^{-1}(\alpha)$ for $\alpha \in (0,\frac{1}{2})$, we note that the density and CDF of $X$ on $[0,\frac{1}{2}]$ are given by
$f(x) = 4 \hsp x$ and $F(x) = 2 \hsp x^2$, respectively.  Hence, $x_\alpha = F^{-1}(\alpha) = \sqrt{\alpha}/\sqrt{2}$ and
\begin{align}  \label{eq:AUC_Y_X_2}
\AUC^{\hsp (\alpha)}(Y, X) = 1 - \frac{\frac{2}{3} \sqrt{2\alpha} + \frac{\alpha}{6}}{1 - \alpha}
\end{align}
for $\alpha \in (0, \frac{1}{2})$, as claimed.  For reasons of symmetry, $\AUC^{(1/2)}(Y,X) = \frac{1}{2}$, whereas $\AUC^{\hsp (\alpha)}(Y, X)$ for $\alpha \in (\frac{1}{2}, 1)$ is obtained from \eqref{eq:AUC_Y_X_1} and \eqref{eq:AUC_Y_X_2}.
\end{proof}

\begin{proof}[Proof of Theorem \ref{thm:AGC_properties}.]
The claims in parts (a), (b), (c), (d), and (f) are immediate from the relations at \eqref{eq:AGC}, \eqref{eq:AGC_granularity}, \eqref{eq:CMA},  and~\eqref{eq:CMA2s}.  We proceed to prove part (e).  First, suppose that $Y = m(X)$ almost surely, where $m$ is nondecreasing.  Then $Y' < Y''$ implies $X' < X''$ almost surely, and we find from~\eqref{eq:CMA2s} that $\AGC(X,Y) = 2 \, \CMA(X, Y) - 1 = 1$.  Conversely, if $\AGC(X,Y) = 1$ arguments as in the proof of Theorem~\ref{thm:AKC_properties}, but now based on the representation~\eqref{eq:CMA2s}, imply that $Y = m(X)$ almost surely, where $m$ is nondecreasing.
\end{proof}

\section{Comparison with Chatterjee correlation}  \label{app:Chatterjee}

\setcounter{equation}{0} 
\setcounter{figure}{0} 

In the nondegenerate population setting, the Chatterjee correlation coefficient $\xi$ of $Y$ on $X$ \citep{Chatterjee2021} is defined as
\begin{align*}
\xi(X,Y) 
= \frac{\int \var(\myE( \one \{ Y \leq y \} \mid X)) \dint G(y)}{\int \var( \one \{ Y \leq y \}) \dint G(y)}.
\end{align*}
Chatterjee's coefficient attains values in the unit interval $[0,1]$, and it holds that $\xi(X,Y) = 0$ if, and only if, $X$ and $Y$ are independent.  Furthermore, Chatterjee's coefficient has the complete predictor property: $\xi(X,Y) = 1$ if, and only if, $Y$ is completely dependent on $X$ in the sense of \citet{Lancaster1963}, i.e., if, and only if, there is a measurable function $g$ such that $Y = g(X)$ almost surely.  

Like the notion of perfect monotone dependence \citep{Kimeldorf1978}, complete dependence fails to be symmetric.  Therefore, reflecting arguments in analogy to those in Remark \ref{re:incompatibility}, Chatterjee's coefficient is asymmetric, with \citet[p.~2010]{Chatterjee2021} noting that 
\begin{quote} 
\small
``this is intentional.  We would like to keep it that way because we may want to understand if $Y$ is a function of $X$, and not just if one of the variables is a function of the other.''  
\end{quote}
However, while $\AGC$ and $\CMA$, and $\AKC$ and $\CID$, are measures of monotone dependence, Chatterjee's coefficient is a measure of complete dependence.  A further difference is that $\AGC$ and $\CMA$, and $\AKC$ and $\CID$, respectively, become symmetric when $X$ and $Y$ have the same 3-granularity and 2-granularity, respectively.  In contrast, it is possible that $\xi(X,Y) \not= \xi(Y,X)$ even when both $X$ and $Y$ are continuous. 

\begin{example}  \label{ex:Chatterjee}
We revisit Example \ref{ex:CMA2AUC} and let $(X,Y)$ be uniform on the triangle with vertices $(0,0)$, $(\frac{1}{2},1)$, and $(1,0)$ in the Euclidean plane.  Then $\xi(X,Y) = \frac{1}{24}$, whereas $\xi(Y,X) = \frac{1}{6}$.
\end{example}

\begin{proof}
The conditional distribution of $Y$ given $X = x$, where $x \in [0,1]$, is uniform on $[0, 1 - 2|x-\frac{1}{2}|]$.  Therefore,
\begin{align*}
\xi(X,Y) 
& = 6 \int_0^1 \var \left( \hsp \myE( \one \{ Y \leq G^{-1}(\alpha) \} \mid X) \right) \dint\alpha \\
& =  6 \int_0^1 \int_0^1 \myP(Y \leq G^{-1}(\alpha) \mid X = x)^2 \, f(x) \dint x \dint \alpha - 2 \\
& = \frac{1}{6},
\end{align*}
where $f(x) = 4x$ for $x \leq \frac{1}{2}$ and $f(x) = 4(1-x)$ for $x > \frac{1}{2}$, and
\begin{align*}
\myP \left( Y \leq G^{-1}(\alpha) \mid X = x \right) = \begin{cases}
\frac{1 - \sqrt{1 - \alpha}}{1 - 2|x - \frac{1}{2}|}, & |x - \frac{1}{2}| < \frac{\sqrt{1 - \alpha}}{2}, \\
1,                                                    & |x - \frac{1}{2}| \geq \frac{\sqrt{1 - \alpha}}{2}.
\end{cases}
\end{align*}
Similarly, the conditional distribution of $X$ given $Y = y$, where $y \in [0,1]$, is uniform on the interval $[\frac{y}{2}, 1 - \frac{y}{2}]$, so
\begin{align*}
\xi(Y,X) 
& = 6 \int_0^1 \var \left( \hsp \myE( \one \{ X \leq F^{-1}(\alpha) \} \mid Y) \right) \dint\alpha \\
& = 6 \int_0^1 \left( \myE \left( \myP \left( X \leq F^{-1}(\alpha) \mid Y \right)^2 \right) - \alpha^2 \right) \dint\alpha \\
& = 6 \int_0^1 \int_0^1 \myP \left( X \leq F^{-1}(\alpha) \mid Y = y \right)^2 \, 2 \hsp (1-y) \dint y \dint\alpha \hsp - \hsp 2 \\
\end{align*}
\begin{align*}
& = 6 \int_0^{1/2} \left( \int_0^{2\sqrt{\alpha/2}} \left( \frac{\sqrt{\alpha/2}-y/2}{1-y} \right)^2 2 \hsp (1-y) \dint y  +\int_{2(1-\sqrt{\alpha/2})}^1 2 \hsp (1-y) \dint y \right) \dint\alpha \\
& \hspace{10mm} + \; 6 \int_{1/2}^1 \int_0^{2\sqrt{(1-\alpha)/2}} \left( \frac{(1-\sqrt{(1-\alpha)/2}) - y/2}{1-y} \right)^2 2 \hsp (1-y) \dint y \dint\alpha \hsp - \hsp 2 \\
& = 6 \left( \frac{37}{288} + \frac{61}{288} \right) - 2 \\
& = \frac{1}{24},
\end{align*}
where we use that
\begin{align*}
\myP \left( X \leq F^{-1}(\alpha) \mid Y = y \right) =
\begin{cases} 
\frac{F^{-1}(\alpha) - y/2}{1 - y},      & 0 \leq y < 2 m_\alpha, \\
\one \{ F^{-1}(\alpha) > \frac{1}{2} \}, & 2 m_\alpha \leq y < 2 M_\alpha, \\
\one \{ F^{-1}(\alpha) < \frac{1}{2} \}, & 2 M_\alpha \leq y \leq 1,
\end{cases} 
\end{align*}
with $m_{\alpha} = \min \{ F^{-1}(\alpha), 1 - F^{-1}(\alpha) \}$, $M_{\alpha} = \max \{ F^{-1}(\alpha), 1-F^{-1}(\alpha) \}$, and $F^{-1}(\alpha) = \sqrt{\alpha/2} \, \one \{ \alpha \leq \frac{1}{2} \} \allowbreak + (1 - \sqrt{(1-\alpha)/2}) \hsp \one \{ \alpha > \frac{1}{2} \}$, respectively.
\end{proof}

Finally, we observe that when $Y_c$ is continuous then $\AGC(X,Y_c)$, $\CMA(X,Y_c)$, and $\xi(X,Y_c)$ admit representations as mixtures of quantities that depend on the law of $(X, \one \{ Y_c \leq G^{-1}(\alpha) \})$ only.  In the case of $\AGC$ and $\CMA$, the representation \eqref{eq:CMA2AUC} holds, where $\AUC^{\hsp (\alpha)}(X,Y_c)$ depends on $(X,Y_c)$ via the law of $(X, \one \{ Y_c \leq G^{-1}(\alpha) \})$ only and vanishes if $X$ and $Y_c$ are independent.  In the case of Chatterjee's correlation, we find that
\begin{align*}
\xi(X,Y_c) = 6 \int_0^1 \var \left( \hsp \myE( \one \{ Y_c \leq G^{-1}(\alpha) \} \mid X) \right) \dint\alpha,
\end{align*}
where the integrand also depends on the law of $(X, \one \{ Y_c \leq G^{-1}(\alpha) \})$ only and vanishes if $X$ and $Y_c$ are independent.  The same statement applies to the closely related, recently proposed integrated $R^2$ measure of \citet{Azadkia2025}.  However, we are unaware of a representation of this kind for $\AKC(X,Y_c)$ and $\CID(X,Y_c)$, respectively.

\section{Proofs for Section~\ref{sec:sample}}  \label{app:proofs_sample}

\setcounter{equation}{0} 
\setcounter{figure}{0} 

\begin{proof}[Proof of Proposition \ref{prop:AKC_empirical}.]
The expression for $\AKC_n$ at \eqref{eq:AKC_n} follows readily from the defining relations at \eqref{eq:AKC} and \eqref{eq:tau}, respectively.  To derive the expression for $\CID_n$ at \eqref{eq:CID2s_n} from \eqref{eq:AKC_n}, we note that 
\begin{align*}
\CID_n = \frac{1}{2} \left( \AKC_n + 1 \right) 
& = \frac{\sum_{i=1}^n \sum_{j=1}^n \left( \one \{ X_i < X_j \} - \one \{ X_i > X_j \} + 1 \right) \one \{ Y_i < Y_j \}}{2 \sum_{i=1}^n \sum_{j=1}^n \one \{ Y_i < Y_j \}} \\
& = \frac{\sum_{i=1}^n \sum_{j=1}^n \left( 2 \cdot \one \{ X_i < X_j \} + \one \{ X_i = X_j \} \right) \one \{ Y_i < Y_j \}}{2 \sum_{i=1}^n \sum_{j=1}^n \one \{ Y_i < Y_j \}} \\
& = \frac{\sum_{i=1}^n \sum_{j=1}^n s(X_i, X_j) \one \{ Y_i < Y_j \}}{\sum_{i=1}^n \sum_{j=1}^n \one \{ Y_i < Y_j \}}, 
\end{align*}
as claimed.
\end{proof}

\begin{proof}[Proof of Proposition \ref{prop:AGC_empirical}.]
For data of the form at \eqref{eq:data}, let $\Fbar_n$ and $\Gbar_n$ denote the empirical MDF of the sample values $X_1, \ldots, X_n$ and $Y_1, \ldots, Y_n$, respectively.  The mid rank $\bar{R}_j$ associated with $Y_j$ equals $\bar{R}_j = \allowbreak \sum_{i=1}^n \one \{ Y_i < Y_j \} \allowbreak + \allowbreak \frac{1}{2} \sum_{i=1}^n \one \{ Y_i = Y_j \} + \frac{1}{2}$.  Therefore, 
\begin{align*}
\Gbar_n(Y_i) = \frac{1}{n} \left( \bar{R}_i - \frac{1}{2} \right), \qquad 
\myE \hsp \left[ \Gbar_n(Y_i) \right] = \frac{1}{2}
\end{align*}
for $i = 1, \ldots, n$, and analogously for $\Fbar_n(X_i) = \frac{1}{n} \hsp \bar{Q}_i$.  Consequently, \eqref{eq:AGC_n} and \eqref{eq:CMA_n} follow from \eqref{eq:AGC} and \eqref{eq:CMA} by plugging in and simplifying.  Similarly, \eqref{eq:CMA2s_n} follows from \eqref{eq:CMA2s} by plugging in and simplifying.  The transition to \eqref{eq:CMA2s_alt_n} then is straightforward.
\end{proof}

\begin{proof}[Proof of Corollary \ref{cor:CPA}.]
If $Y_1, \ldots, Y_n$ are pairwise distinct, we find from \eqref{eq:CMA2s_alt_n} and \eqref{eq:CPA_n} that $\CPA_n = \CMA_n$.  The second equality then follows readily from Theorem 1 in concert with Definition 3 in \citet{Gneiting2022b}.
\end{proof}

Our proofs of the central limit laws in Theorems \ref{thm:AKC_CLT}, \ref{thm:AGC_CLT}, \ref{thm:AKC_CLT_MV}, and \ref{thm:AGC_CLT_MV} rely heavily on classical asymptotic theory for U-statistics \citep{Hoeffding1948} as recently expanded by \citet{Pohle2025a}.  We operate under Scenario \ref{sce:2}, which nests Scenario \ref{sce:1}, and consider the sequence at \eqref{eq:data_MV}. 

Towards the proofs of Theorems \ref{thm:AKC_CLT} and \ref{thm:AKC_CLT_MV}, we recall that 
\begin{align}  \label{eq:AKC_proof}
\AKC^{(k)} = \AKC(X^{(k)},Y) = \frac{\tau^{(k)}}{1 - \twogran_G}
\end{align}
for the population quantities, where $\tau^{(k)}$ is defined at \eqref{eq:tau_k}, and analogously $\AKC^{(k)}_n = \tau^{(k)}_n / (1 - \twogran_n)$ for the empirical quantities, where $k = 1, \ldots, m$.  Proposition A.3 and Lemma B.2 of \citet{Pohle2025a} yield the following limit law, where the entries of the asymptotic covariance matrix are expressed in terms of the kernels $K^{(\operatorname{Kendall}, \, k)}$ and $K^{(2)}$ at \eqref{eq:K_Kendall} and \eqref{eq:K_2}.

\begin{proposition}  \label{prop:AKC_MV}
Under Scenario \ref{sce:2} it holds that
\begin{align*}  
\sqrt{n} \left(
\begin{pmatrix} \tau_n^{(1)} \\ \vdots \\ \tau_n^{(m)} \\ \twogran_n \end{pmatrix} - \begin{pmatrix} \tau^{(1)} \\ \vdots \\ \tau^{(m)} \\ \twogran_G \end{pmatrix} \right) \stackrel{d}{\longrightarrow} \: 
\cN \left( \begin{pmatrix} 0 \\ \vdots \\ 0 \\ 0 \end{pmatrix} , \Sigma_\circ = \left( \Sigma_\circ^{(k,l)} \right)_{k,l=1}^{m+1} \right) ,
\end{align*}
where
\begin{align*}
\Sigma_\circ^{(k,l)} = 4 \, \myE \hspm \left( K^{(\operatorname{Kendall}, \, k)}(X^{(k)},Y) \, K^{(\operatorname{Kendall}, \, l)}(X^{(l)},Y) \right) 
\end{align*}
for\/ $k, l = 1, \ldots, m$,
\begin{align*}
\Sigma_\circ^{(m+1,k)} = \Sigma_\circ^{(k,m+1)} = 4 \, \myE \hspm \left( K^{(\operatorname{Kendall}, \, k)}(X^{(k)},Y) \, K^{(2)}(Y) \right)
\end{align*}
for\/ $k = 1, \ldots, m$, and 
\begin{align*}
\Sigma_\circ^{(m+1,m+1)} = 4 \, \myE \hspm \left( K^{(2)}(Y)^2 \right) ,
\end{align*}
respectively.
\end{proposition}

\begin{proof}[Proof of Theorem \ref{thm:AKC_CLT_MV}]
In view of the defining relations at \eqref{eq:AKC_proof}, the respective relationships for the empirical quantities, and Proposition \ref{prop:AKC_MV}, a routine application of the delta method yields
\begin{align*}
\sqrt{n} \left(
\begin{pmatrix} \AKC_n^{(1)} \\ \vdots \\ \AKC_n^{(m)} \end{pmatrix} - \begin{pmatrix} \AKC^{(1)} \\ \vdots \\ \AKC^{(m)} \end{pmatrix} \right) 
\stackrel{d}{\longrightarrow} \: \cN \! \left( \begin{pmatrix} 0 \\ \vdots \\ 0 \end{pmatrix} , \Sigma = \left( \Sigma^{(k,l)} \right)_{k,l=1}^m \right) ,
\end{align*}
where $\Sigma = \Lambda \Sigma_\circ \Lambda'$ with the matrix $\Sigma_\circ = \left( \Sigma_\circ^{(k,l)} \right)_{k,l=1}^{m+1}$ as specified in Proposition \ref{prop:AKC_MV} and
\begin{align*}
\Lambda = \begin{pmatrix}
		1 / (1 - \twogran_G) & 0 & \cdots & 0 & \tau^{(1)} / (1 - \twogran_G)^2 \\
		0 & 1 / (1 - \twogran_G) & \cdots & 0 & \tau^{(2)} / (1 - \twogran_G)^2 \rule{0mm}{5mm} \\
		\vdots & \vdots & \ddots & \vdots & \vdots \\
		0 & 0 & \cdots & 1 / (1 - \twogran_G) & \tau^{(m)} / (1 - \twogran_G)^2 \end{pmatrix}
\in \real^{m \times (m+1)}.
\end{align*}
By straightforward algebra, 
\begin{align*}
\Sigma^{(k,l)} 
= \frac{\Sigma_\circ^{(k,l)}}{(1 - \twogran_G)^2} + \frac{\tau^{(k)} \, \Sigma_\circ^{(l,m+1)} + \tau^{(l)} \, \Sigma_\circ^{(k,m+1)}}{(1 - \twogran_G)^3} + \frac{\tau^{(k)} \tau^{(l)} \, \Sigma_\circ^{(m+1,m+1)}}{(1 - \twogran_G)^4}
\end{align*}
for $k, l = 1, \ldots, m$, which yields the expression at \eqref{eq:AKC_sigma_kl}. 
\end{proof}

\begin{proof}[Proof of Theorem \ref{thm:AKC_CLT}.]
Evidently, the first part is a special case of Theorem \ref{thm:AKC_CLT_MV}.  If $X$ and $Y$ are independent then it holds that $\tau = 0$ and 
\begin{align*}
\myE \left( K^{(2)}(X,Y)^2 \right) = \frac{1}{9} \left( 1 - \threegran_F \right) \left( 1 - \threegran_G \right)
\end{align*}
by Proposition 4.2 and Corollary 4.5 of \citet{Pohle2025a}, so $\sigma^2$ simplifies to the expression at \eqref{eq:AKC_CLT_independent}.
\end{proof}

Tending now to Theorems \ref{thm:AGC_CLT} and \ref{thm:AGC_CLT_MV}, we note that  
\begin{align*}
\AGC^{(k)} 
= \AGC(X^{(k)},Y) 
= \frac{\rhoG^{\, (k)}}{1 - \threegran_G}
\end{align*}
for the population quantities and analogously $\AGC^{(k)}_n = \rhoG_n^{\, (k)} / (1 - \threegran_n)$ for the empirical quantities, where $k = 1, \ldots, m$.  The empirical quantities are asymptotically equivalent to the U-statistic in eq.~(12) of \citet{Pohle2025a}.  Therefore, Proposition A.3 and Lemma B.2 of \citet{Pohle2025a} yield the following result, where the entries of the asymptotic covariance matrix are expressed in terms of the kernels $K^{(\operatorname{Spearman}, \, k)}$ and $K^{(3)}$ at \eqref{eq:K_Spearman} and \eqref{eq:K_3}, respectively.

\begin{proposition}  \label{prop:AGC_MV}
Under Scenario \ref{sce:2} it holds that
\begin{align*}  
\sqrt{n} \left(
\begin{pmatrix} \rhoG_n^{\, (1)} \\ \vdots \\ \rhoG_n^{\, (m)} \\ \threegran_n \end{pmatrix} - \begin{pmatrix} \rhoG^{\, (1)} \\ \vdots \\ \rhoG^{\, (m)} \\ \threegran_G \end{pmatrix} 
\right) \stackrel{d}{\longrightarrow} \: \cN \left( \begin{pmatrix} 0 \\ \vdots \\ 0 \\ 0 \end{pmatrix} , \Sigma_\circ = \left( \Sigma_\circ^{(k,l)} \right)_{k,l=1}^{m+1} 
\right) ,
\end{align*}
where
\begin{align*}
\Sigma_\circ^{(k,l)} = 9 \, \myE \hspm \left( K^{(\operatorname{Spearman}, \, k)}(X^{(k)},Y) \, K^{(\operatorname{Spearman}, \, l)}(X^{(l)},Y) \right) 
\end{align*}
for\/ $k, l = 1, \ldots, m$,
\begin{align*}
\Sigma_\circ^{(m+1,k)} = \Sigma_\circ^{(k,m+1)} = 9 \, \myE \hspm \left( K^{(\operatorname{Spearman}, \, k)}(X^{(k)},Y) \, K^{(3)}(Y) \right)
\end{align*}
for\/ $k = 1, \ldots, m$, and 
\begin{align*}
\Sigma_\circ^{(m+1,m+1)} = 9 \, \myE \hspm \left( K^{(3)}(Y)^2 \right) ,
\end{align*}
respectively.
\end{proposition}

With Proposition \ref{prop:AGC_MV} at hand, the proofs of Theorems \ref{thm:AGC_CLT_MV} and \ref{thm:AGC_CLT}, respectively, are entirely analogous to the proofs of Theorems \ref{thm:AKC_CLT_MV} and \ref{thm:AKC_CLT}, respectively.

\bibliographystyle{apalike}
\bibliography{newbiblio}

\end{document}